\documentclass[lettersize,journal]{IEEEtran}
\usepackage{amsmath,amssymb,amsfonts}
\usepackage{algorithmic}
\usepackage{algorithm}
\usepackage{array}
\usepackage[caption=false,font=normalsize]{subfig}
\usepackage[table]{xcolor}
\usepackage{textcomp}
\usepackage{stfloats}
\usepackage{url}
\usepackage{verbatim}
\usepackage{graphicx}
\usepackage{amsthm}
\usepackage{xcolor}
\usepackage{hyperref}
\newtheorem{proposition}{Proposition}

\usepackage{amsthm}
\usepackage{cite}
\usepackage{array, makecell}
\usepackage{mathtools}
\newtheorem{theorem}{Theorem}
\newtheorem{lemma}{Lemma}
\newtheorem{assumption}{Assumption}
\usepackage{booktabs}
\usepackage{placeins}     

\begin{document}

\title{REVERB-FL: Server-Side Adversarial and Reserve-Enhanced Federated Learning for Robust Audio Classification}

\author{Sathwika Peechara and Rajeev Sahay

\thanks{S. Peechara is with the Department of Computer Science and Engineering, University of California San Diego, San Diego, CA, 92093 USA. E-mail: speechara@ucsd.edu.}
\thanks{R. Sahay is with the Department of Electrical and Computer Engineering, University of California San Diego, San Diego, CA, 92093 USA. E-mail: r2sahay@ucsd.edu.}
\thanks{This work was supported in part by the UC San Diego Academic Senate under grant RG116457 and in part by the National Science Foundation (NSF) under grant ECCS-2512912.}}
\maketitle

\begin{abstract} 
Federated learning (FL) enables a privacy-preserving training paradigm for audio classification but is highly sensitive to client heterogeneity and poisoning attacks, where adversarially compromised clients can bias the global model and hinder the performance of audio classifiers. To mitigate the effects of model poisoning for audio signal classification, we present \textbf{REVERB-FL}, a lightweight, server-side defense that couples a small \emph{reserve set} (approximately 5\%) with pre- and post-aggregation retraining and adversarial training. After each local training round, the server refines the global model on the reserve set with either clean or additional adversarially perturbed data, thereby counteracting non-IID drift and mitigating potential model poisoning without adding substantial client-side cost or altering the aggregation process. We theoretically demonstrate the feasibility of our framework, showing faster convergence and a reduced steady-state error relative to baseline federated averaging. We validate our framework on two open-source audio classification datasets with varying IID and Dirichlet non-IID partitions and demonstrate that REVERB-FL mitigates global model poisoning under multiple designs of local data poisoning. 
\end{abstract}


\section{Introduction}

\IEEEPARstart{A}{udio} classification tasks in machine learning are ubiquitous in many applications, including speech recognition, environmental sound classification, and sentiment analysis. Deep neural networks trained on spectrogram-based features, such as Mel-frequency cepstral coefficient (MFCC) features or short-time Fourier transforms (STFT), have shown to be effective in extracting time-frequency patterns from audio signals. However, these models often require centralized large scale training data, raising practical concerns. Audio recordings are inherently sensitive, often containing biometric or contextual information, and centralized aggregation risks privacy violations. Furthermore, audio data is naturally distributed across edge devices (e.g., smartphones, IoT sensors, etc.), making centralized storage inefficient.

Federated learning (FL) allows for a privacy-preserving training approach across edge devices without transmitting raw data away from device, making it a more secure method for speech and audio applications, where data cannot be centralized due to privacy regulations or bandwidth constraints\cite{2018federated}. In an FL network, each client (i.e., edge device) locally trains a model on its local collected and stored dataset using stochastic gradient descent (SGD). Model parameters, rather than the data itself, are subsequently aggregated to a central server using weighted averaging (FedAvg)\cite{fedAvg_baseline}. Despite its privacy-preserving nature, audio FL frameworks are susceptible to \emph{model poisoning attacks}, where malicious clients inject adversarial perturbations (e.g., via gradient-based attacks, or additive noise) into their local training data before computing updates \cite{bhagoji2019analyzing,Rodriguez-Barroso2022Survey,fl_poisoning_attack}. This creates perturbations that are \emph{audio-wise imperceptible}, (i.e., they do not noticeably alter the audio to human listeners) \cite{robust_audio_adversarial} but hinder the training process, resulting in backdoor attacks or poor convergence, which ultimately compromise the final model's performance.

Robustness against adversarial model poisoning attacks is a major concern in FL-based audio classification \cite{fl_survey}. Specifically, prior work \cite{ evasion_attacks_fl} has shown that audio classification models are particularly vulnerable to model poisoning attacks due to the high dimensionality of spectrograms and the sensitivity of deep learning models to input perturbations. Even small perturbations can cause significant accuracy drops in applications such as voice assistants and environmental sound detection, where misclassifications can have safety or operational consequences. These challenges motivate the need for methods that both preserve privacy and improve robustness when applying FL to audio domains. Although prior studies have explored adversarial robustness in centralized audio models or in limited non-audio classification-based FL applications \cite{fl_audio_challenges, robust_audio_adversarial}, lightweight defenses that systematically address poisoning robustness while maintaining convergence stability in spectrogram-based federated audio classification have not been investigated. 

To address these challenges, we propose Reserve-Enhanced Verification and Robustness in Federated Learning (\textbf{REVERB-FL}), which is a novel framework that combines server-side stabilization with adversarial robustness, designed to both improve convergence and mitigate training-time attacks. Our approach is motivated by the observation that federated training amplifies instability under heterogeneous data and adversarial conditions, and that existing filtering-based defenses often rely on unrealistic trust assumptions (i.e., the ability to reliably distinguish adversarial updates from honest-but-heterogeneous ones under non-IID data) \cite{krum, trimmedmean, bulyan, median}. Such approaches reduce the effective training data available to the global model, which is particularly problematic in audio domains where every sample is valuable due to the difficulty of collecting and cleaning high-quality data. By contrast, our methods strengthen the global model directly at the server, providing a defense mechanism that is lightweight, scalable, and compatible with standard FL pipelines \cite{server_finetune, defense_aware_fl}. We provide a theoretical convergence bound highlighting the feasibility of our method, which is empirically validated in multiple adversarial settings and shown to outperform multiple considered baselines. 

\textbf{Summary of Contributions:} 
Specifically, our contributions can be summarized as follows:
\begin{enumerate} 
    \item \textbf{Federated audio robustness framework with reserve set retraining:} We develop a federated audio classification framework with STFT-based inputs that incorporates server-side reserve set retraining to stabilize model aggregation in adversarial environments. 

    \item \textbf{Convergence analysis:} We provide a theoretical convergence bound for our proposed framework, showing that our framework maintains the standard guarantees of FedAvg while enhancing robustness in adversarial settings.


    \item \textbf{Empirical evaluation}: We provide extensive experiments demonstrating that both reserve set retraining and adversarial retraining significantly improve robustness compared to multiple baselines, including the standard FedAvg approach, highlighting the efficacy of our method for secure federated audio systems.
\end{enumerate}

The remainder of this paper is organized as follows. Sec.~\ref{related_work} reviews related work on adversarial robustness in federated learning-based audio classification. Sec.~\ref{methods} presents our signal and classifier modeling, formalizes the federated learning setup, defines the threat model for adversarial poisoning attacks, and details our proposed REVERB-FL framework. Sec.~\ref{methods} also presents the theoretical feasibility of our proposed framework. Sec.~\ref{results} describes our experimental setup and presents our empirical evaluation across multiple datasets, data partitions, and attack scenarios. Finally, Sec.~\ref{conclusion_sec} discusses concluding remarks and directions for future work.

\section{Related Works} \label{related_work}
Deep learning methods have achieved strong performance on audio classification tasks by operating on time–frequency representations such as STFTs or mel-spectrograms and training convolutional neural network (CNN) or hybrid CNN–recurrent neural network (RNN) architectures, with more recent work exploring transformers and self-supervised representations \cite{spectrograms, zhang2022transformers, wang2022selfsupervised, dai2017very}. However, centralized training of these models raises privacy concerns, as audio recordings often contain sensitive biometric or contextual information that cannot be easily anonymized.

Federated learning (FL) addresses these privacy concerns by enabling collaborative model training without centralizing raw audio data, preserving data privacy. Applications span speech recognition, keyword spotting, and emotion recognition, where data are naturally distributed across devices and cannot always be centralized \cite{fl_speech_emotion, fl_kws, fl_speech_recognition}. Canonical FedAvg \cite{fedAvg_baseline} aggregates client updates by data size, but non-independent and identically distributed (non-IID) 
client data (where label distributions vary across clients) and device heterogeneity make audio FL particularly challenging \cite{hsu2019measuring, fl_audio_challenges}. Numerous FL variants have been proposed to stabilize convergence or personalize models, including FedProx, Per-FedAvg, and SCAFFOLD \cite{fedprox, per_fedavg, scaffold}. Yet, these methods primarily address heterogeneity rather than robustness.

Meanwhile, adversarial vulnerability has been widely documented in both centralized and federated settings. Gradient-based attacks \cite{fgsm, pgd}, as well as additive Gaussian noise can significantly reduce audio model performance in centralized models at inference \cite{robust_audio_adversarial}. In FL, these same perturbation techniques are applied as \emph{model poisoning attacks}, where malicious clients perturb local training data, and the poisoned updates propagate through aggregation and bias the global model \cite{bhagoji2019analyzing,poison_fl}. The effect is amplified under non-IID conditions, making audio FL an attractive target for adversaries.

Defenses against model poisoning in federated learning have been extensively studied for general FL settings, though few target audio-specific applications exist. One recent exception is the Knowledge Distillation Defense Framework (KDDF) for federated automatic speech recognition\cite{knowledge_distillation}, which defends against backdoor attacks by detecting triggers at inference time, but does not address general gradient-based poisoning attacks applied during training. 

The existing general (non-audio specific) FL defenses can be grouped into three main strategies. Robust aggregation functions such as Krum, trimmed mean, median, Bulyan, and more recent dynamic schemes \cite{krum, bulyan, trimmedmean, median} aim to suppress malicious updates but risk discarding useful data when clients are few. Regularization-based approaches such as FedProx, SCAFFOLD, or personalized FL frameworks \cite{fedprox, scaffold, pfedme} mitigate client drift but do not directly improve adversarial robustness. Adversarial training \cite{tramer2018ensemble} has been extended to FL, but client-side variants can be unstable under non-IID and add device overhead \cite{fl_at_instability, yan2025federated}.
Server-side defenses have therefore gained attention, including finetuning on trusted side datasets \cite{server_finetune} and defense-aware aggregation \cite{defense_aware_fl}. Complementary to these, Yan \textit{et al.} propose a federated adversarial training scheme that generates gradient-based adversarial examples on clients and couples this with a personalized evaluation policy to reweight aggregation, improving robustness under both multiple threat models \cite{yan2025federated}. Although their setup involves multi-site datasets with inherently heterogeneous distributions, the method does not explicitly address or analyze non-IID client drift, and its reliance on local adversarial training increases on-device cost. 

In contrast, our framework shifts robustness entirely to the server via reserve-set adversarial retraining, providing non-IID stability without modifying client updates or aggregation. Building on these insights, our work contributes a server-side defense specifically for audio FL. We introduce REVERB-FL, which combines reserve-set retraining to stabilize global updates with adversarial augmentation to strengthen robustness against poisoning. Unlike aggregation-based defenses, REVERB-FL does not discard client updates, and unlike client-side adversarial training, it introduces no additional cost or instability at the device level. 

We further clarify the distinction from existing server-side approaches. Prior work on label-flipping and backdoor trigger attacks \cite{knowledge_distillation} relies on detecting anomalous patterns at inference time, an approach that is orthogonal to ours, as trigger-based attacks introduce detectable input patterns whereas the gradient-based adversarial perturbations we consider are imperceptible and harder to flag. Server-side fine-tuning approaches \cite{server_finetune} apply trusted-data refinement primarily for performance improvement under non-IID conditions, not adversarial robustness. Defense-aware aggregation methods \cite{defense_aware_fl} modify the aggregation rule itself (e.g., filtering weights, using trimmed mean, etc.), whereas REVERB-FL operates between communication rounds without altering aggregation, making it compatible with any existing FL pipeline. To our knowledge, this is the first framework combining server-side reserve-set retraining with adversarial augmentation specifically for robust federated audio classification, supported by both empirical results and a convergence analysis that extends FedAvg guarantees under adversarial conditions.

\section{Methodology} \label{methods}
In this section, we first formalize our signal model (Sec.~\ref{subsec:signalmodel}) and the classifier architecture (Sec.~\ref{subsec:classifier}). We then describe the federated learning protocol (Sec.~\ref{subsec:federatedsetup}) and define the adversarial threat model for model poisoning attacks (Sec.~\ref{subsec:attacks}). Finally, we detail the REVERB-FL defense framework (Sec.~\ref{subsec:framework}), and provide a theoretical convergence analysis (Sec.~\ref{subsec:conv_an}).
\subsection{Signal Modeling}\label{subsec:signalmodel}
Audio signals are represented in the time-frequency domain to capture both spectral and temporal characteristics relevant to classification \cite{stft}.
Each utterance (or audio clip) is modeled as a discrete-time waveform, $\mathbf{x}$, where
$x[n]$ is the $n^{\text{th}}$ time sample of $\mathbf{x}$, at sampling rate $f_s$.  
Its short-time Fourier transform (STFT) is computed as
\begin{equation}
  X(\nu, \tau) = \sum_{n=-\infty}^{\infty} x[n]\,
  w[n-\tau]\,e^{-j 2\pi \nu n / F},
  \label{eq:stft}
\end{equation}
where $w[n]$ is a window function of length $L_w$ applied to each sample, $\tau$ indexes the frame, $\nu$ indexes the frequency bin, and $F$ is the fast Fourier transform (FFT) size.
Each STFT frame thus provides a localized spectral snapshot of the waveform.

The complex spectrogram is split into real and imaginary components, $\mathbf{X}(\cdot, \cdot, 1) = \text{Re}(X(\nu,\tau))$ and 
$\mathbf{X}(\cdot, \cdot, 2) = \text{Im}(X(\nu,\tau))$ (for compatibility with real-valued networks), 
forming a three-dimensional tensor
$\mathbf{X} \in \mathbb{R}^{n_f \times T \times 2}$,
where $n_f$ and $T$ denote frequency bins and time frames, respectively. While we adopt STFT-based representations following standard practice in audio FL~\cite{spectrograms, stft}, alternative representations such as perception-inspired or interpretable acoustic features~\cite{auditory_perception_disease, heart_sound_interpretable} may offer different robustness and interpretability tradeoffs. Since REVERB-FL operates on model parameters and the reserve set rather than the input representation directly, our framework can be generally applied to other domains and can transfer to other feature spaces without modification to the proposed defense mechanism.

We denote the global label set by $\mathcal{Y}=\{1,\dots,K\}$,  where $K$ is the number of classes.
Client $n$ holds a local dataset
$\mathcal{D}_n=\{(\mathbf{X}_i,y_i)\}_{i=1}^{D_n}$, with $D_n = |\mathcal{D}_n|$ samples, where data may follow client-specific or non-identical distributions.
This formulation accommodates both independent and identically distributed (IID) and non-IID label partitions, instantiated through random or Dirichlet label-skew splits \cite{dirichlet}, further discussed in
Sec.~\ref{sec:exp_setup}.

\subsection{Classifier Modeling}\label{subsec:classifier}
We adopt a deep neural network architecture widely used in spectrogram-based recognition tasks.
Let
$f(\cdot;\theta): \mathbb{R}^{n_f \times T \times 2} \rightarrow \mathbb{R}^K$
$y \in \{0,1\}^K$
denote the classifier parameterized by weights $\theta$, mapping each input
tensor $\mathbf{X}$ to a $K$-dimensional output vector representing class scores.
For a training pair $(\mathbf{X}, y)$ with one-hot encoded label
$y\!\in\!\{0,1\}^{K}$,
the cross-entropy loss is defined as
\begin{equation}
  \ell(\mathbf{X}, y;\theta) = -\sum_{k=1}^{K} y_k
  \log f_k(\mathbf{X};\theta),
  \label{eq:loss}
\end{equation}
where $f_k(\mathbf{X};\theta)$ denotes the predicted probability for class $k$, obtained by applying softmax to the network's output logits.

Training minimizes the empirical risk $\tfrac{1}{D_n}\sum_{(\mathbf{X},y)\in\mathcal{D}_n}\ell(\mathbf{X},y;\theta)$ at each client using stochastic gradient descent with regularization techniques such as weight decay and dropout. 

\subsection{Federated Learning Setup}\label{subsec:federatedsetup}
We operate under a closed-set classification assumption, in which the global label space $\mathcal{Y} = \{1,\dots,K\}$ is fixed and shared across all clients and the server-side reserve set $\mathcal{D}_r$; extensions to open-set scenarios are discussed in Sec.~\ref{sec:limitations}. This is consistent with standard federated audio classification settings \cite{fl_audio_challenges}, where the target class vocabulary is determined prior to deployment. We consider two data distribution scenarios: \emph{independent and identically distributed (IID)}, where all FL client's local label distribution matches the global distribution and \emph{non-IID}, where label distributions vary significantly across clients, such as only having samples from a subset of classes, or highly unbalanced class proportions. Non-IID partitions are common in audio FL due to speaker-specific devices or geographic clustering of sounds. We instantiate non-IID splits using Dirichlet label-skew with concentration $\alpha$ \cite{dirichlet}, evaluated in Sec.~\ref{results}.

Let $\mathcal{D}$ denote the complete dataset distributed across all clients. Before federated training begins, the server collects a stratified reserve set $\mathcal{D}_r$ (approximately 5\% of the total available data) by sampling from the clients -- a common practice in non-IID FL \cite{zhao2018federated}. Specifically, for each class $k \in \{1,\ldots,K\}$, we sample approximately $5\%$ of all instances of class $k$ uniformly at random across all clients, ensuring $\mathcal{D}_r$ maintains the global class distribution. This stratified sampling ensures the reserve objective closely approximates the global objective, yielding small approximation error in the convergence analysis (Sec.~\ref{subsec:conv_an}). The sampled data is transmitted to the server once before training and removed from the clients' local datasets. We denote the remaining local dataset at client $n$ after this removal as $\mathcal{D}_n$, ensuring $\mathcal{D}_r \cap \mathcal{D}_n = \emptyset$ for all $n$. This size is chosen to balance two competing considerations: a larger reserve set improves approximation of the global objective (reducing $\varepsilon_r$ in Proposition~\ref{prop:reserve}) but increases the one-time communication cost of transmitting data to the server. At 5\%, the payload is small enough to preserve client privacy and avoid compromising local dataset sizes, consistent with the shared data fraction used in \cite{zhao2018federated} to stabilize non-IID convergence. Since $\mathcal{D}_r$ is stratified and representative by construction, its quality can be validated prior to deployment; and because it is collected once before training begins, any mild imbalance has bounded impact on the mismatch term $\varepsilon_r$. Increasing $|\mathcal{D}_r|$ would improve performance at the cost of reducing privacy and incurring additional communication overhead (for transmitting raw data) until all data has been transmitted to the global server at which point the paradigm will shift to a centralized training scenario. Conversely, reducing $|\mathcal{D}_r|$ would decrease overall performance until $|\mathcal{D}_r| = 0$ at which point REVERB-FL would not be viable.

We consider a standard synchronous federated learning (FL) framework with a central server and $N$ distributed clients, indexed by $n \in \{1,\dots,N\}$. Each client $n$ possesses a local dataset
$\mathcal{D}_n = \{(\mathbf{X}_i, y_i)\}_{i=1}^{D_n}$, where $D_n = |\mathcal{D}_n|$ denotes the number of samples at client $n$, and the global objective function is defined as
\begin{equation}
  \varphi(\theta) = \frac{1}{N} \sum_{n=1}^{N} \varphi_n(\theta), \quad
  \varphi_n(\theta) = \mathbb{E}_{(\mathbf{X},y)\sim\mathcal{D}_n}
  [\ell(\mathbf{X}, y; \theta)],
  \label{eq:global_obj}
\end{equation}
where $\ell(\mathbf{X},y;\theta)$ is the cross-entropy loss in
\eqref{eq:loss}.  
The FL objective is to minimize $\varphi(\theta)$ without centralizing the data over $R$ communication rounds.

At each communication round $t$, the server samples a subset of
$m$ clients $S_t \subseteq \{1,\dots,N\}$ uniformly without replacement.
The selected clients receive the current global model
$\theta^{(t)}$ and perform $\tau$ local steps of stochastic gradient descent
(SGD) updates using their own data:
\begin{equation}
  \theta_{n}^{(t,j+1)} = \theta_{n}^{(t,j)} -
  \eta\,\nabla_{\theta}\ell(\mathbf{X}_{n}^{(t,j)}, y_{n}^{(t,j)}; \theta_{n}^{(t,j)}),
  \label{eq:local_sgd}
\end{equation}
for local step $j \in \{0,\dots,\tau-1\}$ and step size $\eta$.
After $\tau$ local updates, the client transmits its parameters
$\theta_{n}^{(t,\tau)}$ back to the server.

The server aggregates the received models using data-size weighted averaging
(FedAvg) \cite{fedAvg_baseline}:
\begin{equation}
  \theta^{(t+1)} =
\sum_{n \in S_t} \frac{D_n}{\sum_{k \in S_t} D_k}\,\theta_{n}^{(t,\tau)}.
  \label{eq:fedavg}
\end{equation}
This update rule is unbiased when clients are sampled uniformly at random, regardless of data distribution, and remains a standard
baseline in FL literature \cite{li2020convergencefedavgnoniiddata,wang2020tacklingobjectiveinconsistencyproblem}.

To improve stability, the server then performs additional reserve-set retraining
using a small trusted subset
$\mathcal{D}_{r} \subset \mathcal{D}$ of size $|\mathcal{D}_{r}|/|\mathcal{D}| \approx 5\%$.
After aggregation, $\theta^{(t+1)}$ is refined by $r$ SGD steps on
$\mathcal{D}_{r}$ before being broadcast to all clients in the next round.
This step mitigates the drift induced by non-IID client updates or adversarial
poisoning and preserves convergence guarantees, as analyzed in
Sec.~\ref{subsec:conv_an}. This setup follows standard synchronous FL protocols
\cite{scaffold,fl_survey}.

\subsection{Adversarial model poisoning}\label{subsec:attacks}
We adopt an untargeted training-time poisoning threat model in the STFT feature space. Given a sample $\mathbf{X} \in \mathcal{X}$, where $\mathcal{X}$ is the admissible set clipped element-wise to remain a valid audio signal, with label $y$ and current parameters $\theta$, the adversary seeks a perturbation $\boldsymbol{\delta}$ such that the perturbed example $\tilde{\mathbf{X}}=\mathbf{X}+\boldsymbol{\delta}$ maximizes the loss under an $\ell_\infty$ budget and feasibility constraint. Let $\mathbf{X}'$ denote a candidate perturbed input. Formally, we seek to find
\begin{equation}
  \tilde{\mathbf{X}} \;=\; \arg\max_{\mathbf{X}'\in\mathcal{X},\,\|\mathbf{X}'-\mathbf{X}\|_\infty\le\varepsilon}\;
  \ell(\mathbf{X}',y;\theta), 
  \label{eq:adv-problem}
\end{equation}
where $\varepsilon$ is the $\ell_\infty$ budget of the perturbation. However, \eqref{eq:adv-problem} is a highly non-linear optimization problem without an exact solution. Thus, we approximate a solution to \eqref{eq:adv-problem} using three potent gradient-based perturbation methods, originally developed for test-time evasion attacks \cite{fgsm, pgd}, which we adapt for training-time model poisoning.
Unless stated otherwise, all attacks are untargeted and operate on spectrogram tensors (not waveforms).

\paragraph*{FGSM (single-step)}
The fast gradient sign method performs one signed gradient ascent step on the input \cite{fgsm}:
\begin{equation}
  \tilde{\mathbf{X}} = \mathbf{X} + \varepsilon\,\mathrm{sign}\big(\nabla_{\mathbf{X}}\,\ell(\mathbf{X},y;\theta)\big),
  \label{eq:fgsm}
\end{equation}
where the result is clipped element-wise to the admissible set $\mathcal{X}$.

\paragraph*{PGD (iterative)}
Projected gradient descent applies $I$ perturbation iterations with step size
$\varepsilon/I$ and projection back to the $\ell_\infty$ ball \cite{pgd}. Formally, PGD is given by
\begin{equation}
  \tilde{\mathbf{X}}^{(i+1)} \;=\;
  \Pi_{\mathcal{B}_\varepsilon(\mathbf{X})\cap\mathcal{X}}
  \Big(\tilde{\mathbf{X}}^{(i)} + \frac{\varepsilon}{I}\,\mathrm{sign}\big(\nabla_{\mathbf{X}}\,\ell(\tilde{\mathbf{X}}^{(i)},y;\theta)\big)\Big),
  \label{eq:pgd}
\end{equation}
where $\tilde{\mathbf{X}}^{(0)}$ is initialized with a random start 
$\tilde{\mathbf{X}}^{(0)}=\mathbf{X}+\mathbf{u}$, where each element of $\mathbf{u}$ is sampled uniformly from $[-\varepsilon,\varepsilon]$.

\paragraph*{AWGN (stochastic corruption)}
Additive Gaussian white noise models environmental perturbations:
\begin{equation}
  \tilde{\mathbf{X}} = \mathbf{X} + \boldsymbol{n},\qquad \boldsymbol{n}\sim\mathcal{N}(0,\sigma^2 I),
  \label{eq:awgn}
\end{equation}
where $\boldsymbol{n}$ is Gaussian noise and the result is clipped element-wise to $\mathcal{X}$.

These attacks have been shown to significantly degrade performance in both centralized audio classification and federated settings \cite{robust_audio_adversarial, evasion_attacks_fl}.
We simulate training-time poisoning by designating a fixed adversarial set $A \subseteq \{1,\dots,N\}$ with $|A| = \lceil \rho N \rceil$, where $\rho$ is the adversarial fraction. In round $t$, when the server samples $m$ clients $S_t$, the fraction of adversarial clients among the $m$ sampled clients is $\beta_t = |S_t \cap A|/m$. These adversarial clients apply perturbations to their local training data before computing local updates. Perturbations are crafted using the client's current local model $\theta_n^{(t,j)}$ during local training. Labels remain unchanged, and the poisoned updates propagate through FedAvg aggregation to bias the global model (variant used is specified in Sec.~\ref{sec:exp_setup}). We note that our considered threat model focuses on the scenario in which adversarial clients poison in-distribution samples with in-distribution labels in an attempt to corrupt the training process. Alternatively, adversarial clients may instead inject out-of-distribution samples with in-distribution labels into the training process to induce training corruption. Such a scenario represents an open-set variant outside the scope of this work, discussed further in Sec.~\ref{sec:limitations}.

\subsection{Proposed Framework: REVERB-FL}\label{subsec:framework}
We propose Reserve-Enhanced Verification and Robustness in Federated Learning (\textbf{REVERB-FL}), a
lightweight server-centric defense framework designed to improve robustness of
federated audio models against any data-level perturbations injected during training-time. It is a general framework for any audio FL network containing zero or more adversarially poisoned clients whose data are poisoned according to the form $\tilde{\mathbf{X}}=\mathbf{X}+\boldsymbol{\delta}$ (where $\boldsymbol{\delta}$ can, but does not have to be, a gradient-based attack as reflected by the different types of poisoning attacks considered in Sec. \ref{subsec:attacks}).
Rather than modifying the aggregation rule or imposing costly client-side adversarial
training, REVERB-FL reinforces the global model after aggregation through two
complementary mechanisms: (i) \emph{reserve-set pretraining and retraining} on a small, trusted subset
of clean data and (ii) \emph{adversarial augmentation} of this reserve set using
gradient-based perturbations.

\paragraph*{(i) Reserve-set pretraining and retraining}
Let $\mathcal{D}_r$ denote the stratified reserve set (approximately 5\%) held at the server, collected by sampling from clients before federated training begins, with sampled data removed from client datasets to ensure there is no overlap. 
First, the global model is pretrained on the reserve set $\mathcal{D}_r$ for a small number of epochs before federated training. Then, after FedAvg aggregation in each round $t$ produces $\theta^{(t+1,0)}$, the server performs $r$ additional SGD steps on $\mathcal{D}_r$ given by
\begin{equation}
  \theta^{(t+1,s)} \leftarrow \theta^{(t+1,s-1)} - \eta_r \nabla_{\theta}\ell(\mathbf{X}_r^{(s)},y_r^{(s)};\theta^{(t+1,s-1)}),\label{eq:reserve-update}
\end{equation}
where $s=1,\ldots,r$, $(\mathbf{X}_r^{(s)},y_r^{(s)})\in\mathcal{D}_r$ are mini-batches, $\eta_r$ is the server learning rate, and the final model $\theta^{(t+1)} = \theta^{(t+1,r)}$ is broadcast to clients. Here $r$ corresponds to one epoch over $\mathcal{D}_r$ (see Sec.~\ref{sec:exp_setup}). This reserve update corrects bias accumulated from non-IID or poisoned client updates and provides an additional descent step that stabilizes convergence.

\paragraph*{(ii) Adversarial augmentation}
To further enhance robustness, the server generates adversarial variants $\tilde{\mathbf{X}}_r$ of reserve inputs $\mathbf{X}_r$ using the attacks defined in Sec.~\ref{subsec:attacks}. Each clean example $(\mathbf{X}_r,y_r)$ in the mini-batch is augmented with its adversarial variant $(\tilde{\mathbf{X}}_r,y_r)$ before reserve retraining, where $\tilde{\mathbf{X}}_r$ is generated using the attacks from Sec.~\ref{subsec:attacks}. This augmentation implicitly regularizes the model to maintain correct predictions within the perturbation neighborhood, providing adversarial invariance at the aggregation level without altering client-side computation. We evaluate reserve retraining with clean data only (Retrain (No Poison)), single-attack augmentation (Retrain (FGSM), Retrain (PGD), Retrain (AWGN)), and mixed augmentation (Retrain (All Adversarial)). Furthermore, note that at the client-level, the poisoning method can vary by client and between rounds. Thus, at the server, we do not make any assumptions about which clients trained on poisoned data or subsequently the poisoning approach used. Our complete training framework of REVERB-FL is summarized in Algorithm~\ref{alg:reverb}. While a fully adaptive adversary aware of the reserve-set mechanism is not evaluated here, robustness under unknown and varying adversarial strategies is assessed in Sec.~\ref{results} through two complementary evaluations: (i) cross-attack evaluation, where each retraining configuration is tested against mismatched attack types (e.g., Retrain (FGSM) evaluated under PGD poisoning), and (ii) mixed-attack evaluation, where each adversarial client randomly applies one of FGSM, PGD, or AWGN per round, simulating a heterogeneous adversarial population with no fixed strategy.

\begin{algorithm}[t]
\caption{REVERB-FL Training Protocol}
\label{alg:reverb}
\begin{algorithmic}[1]
\STATE \textbf{Input:} Reserve set $\mathcal{D}_r$, clients $\{1,\dots,N\}$ with datasets $\{\mathcal{D}_n\}$, rounds $R$, local steps $\tau$, reserve steps $r$
\STATE Pretrain $\theta^{(0)}$ on $\mathcal{D}_r$ for 3 epochs (see Sec.~\ref{sec:exp_setup})
\FOR{$t = 0$ to $R-1$}
    \STATE Server samples $m$ clients $S_t \subseteq \{1,\dots,N\}$
    \STATE Server broadcasts $\theta^{(t)}$ to clients in $S_t$
    \FOR{each client $n \in S_t$ \textbf{in parallel}}
        \STATE $\theta_n^{(t,0)} \leftarrow \theta^{(t)}$
        \FOR{$j = 0$ to $\tau-1$}
            \STATE Sample minibatch $(\mathbf{X}_n, y_n)$ from $\mathcal{D}_n$
            \STATE $\theta_n^{(t,j+1)} \leftarrow \theta_n^{(t,j)} - \eta \nabla_{\theta} \ell(\mathbf{X}_n, y_n;\theta_n^{(t,j)})$
        \ENDFOR
        \STATE Send $\theta_n^{(t,\tau)}$ to server
    \ENDFOR
    \STATE Server aggregates: $\theta^{(t+1, 0)} \leftarrow \sum_{n \in S_t} \frac{D_n}{\sum_{k \in S_t} D_k} \theta_n^{(t,\tau)}$
    \FOR{$s = 1$ to $r$}
        \STATE Sample minibatch $(\mathbf{X}_r^{(s)}, y_r^{(s)})$ from $\mathcal{D}_r$ (with optional adversarial augmentation)
        \STATE $\theta^{(t+1,s)} \leftarrow \theta^{(t+1,s-1)} - \eta_r \nabla_{\theta}\ell(\mathbf{X}_r^{(s)},y_r^{(s)};\theta^{(t+1,s-1)})$
    \ENDFOR
    \STATE $\theta^{(t+1)} \leftarrow \theta^{(t+1,r)}$
\ENDFOR
\STATE \textbf{Return:} Final global model $\theta^{(R)}$
\end{algorithmic}
\end{algorithm}

\subsection{Convergence Analysis}\label{subsec:conv_an}
Here, we analyze the convergence of REVERB-FL by deriving a bound that demonstrates faster convergence compared to baseline FedAvg under adversarial conditions. We employ standard assumptions from federated learning literature \cite{fedAvg_baseline, li2020convergencefedavgnoniiddata, wang2020tacklingobjectiveinconsistencyproblem}:

\begin{assumption}[Smoothness] \label{ass:smooth} 
Each local objective $\varphi_n(\theta)$ is $L$-smooth:
\[
\varphi_n(\theta') \le \varphi_n(\theta)
+ \langle \nabla \varphi_n(\theta),\, \theta'-\theta\rangle
+ \tfrac{L}{2}\|\theta'-\theta\|^2.
\]
\end{assumption}

\begin{assumption}[Bounded gradient variance] \label{ass:variance}
For stochastic mini-batch gradients during local SGD,
\[
\mathbb{E}\!\left[\|\nabla \varphi_n(\theta_n^{(t,j)})-\nabla \varphi_n(\theta_n)\|^2\right] \le \sigma_g^2.
\]
\end{assumption}

\begin{assumption}[Bounded client drift] \label{ass:drift}
Across clients,
\[
\mathbb{E}\!\left[\|\nabla \varphi_n(\theta^{(t)})-\nabla \varphi(\theta^{(t)})\|^2\right] \le \zeta^2,
\qquad\]

\[\text{where }~ \varphi(\theta)=\tfrac{1}{N}\sum_{n=1}^{N}\varphi_n(\theta).
\]
\end{assumption}

\begin{assumption}[Strong convexity] \label{ass:convex}
The global objective $\varphi$ is $\mu$-strongly convex:
\[
\varphi(\theta') \ge \varphi(\theta)
+ \langle \nabla \varphi(\theta),\, \theta'-\theta\rangle
+ \tfrac{\mu}{2}\|\theta'-\theta\|^2.
\]
\end{assumption}

\noindent\textbf{Remark.} 
While deep neural networks are generally non-convex, the strong convexity assumption (Assumption \ref{ass:convex}) is standard in federated learning convergence analyses\cite{dinh2022new, wang2020tacklingobjectiveinconsistencyproblem, scaffold} and enables tractable linear convergence rates. This assumption can be interpreted as a local property near stationary points during training or as characterizing favorable gradient descent conditions under which the algorithm moves towards convergence. Importantly, our theoretical analysis is intended to characterize the qualitative behavior of REVERB-FL over the course of training in terms of the contraction rate and steady-state error, rather than making quantitative claims about specific loss values. Even under non-convexity, the two-phase descent structure (FedAvg aggregation followed by reserve retraining) and the damping of adversarial bias through $(1-\mu\gamma_r)^r$ reflect dynamics that are empirically consistent with the observed convergence curves in Sec.~\ref{results}. Thus, even under the non-convexity assumption, we are able to capture the real-world improvement provided by REVERB-FL compared to standard FedAvg as corroborated by our empirical results.

\medskip
\noindent\textbf{Notation.} We denote the effective aggregation step size as $\gamma_g = \eta\tau$ (where $\eta$ is the client learning rate and $\tau$ is the number of local SGD steps from \eqref{eq:local_sgd}) and the reserve step size as $\gamma_r = \eta_r$ (the server learning rate from \eqref{eq:reserve-update}).

\noindent\textbf{Attack model.}
Before training, a fixed adversarial subset $A\subseteq\{1,\dots,N\}$ with $|A|=\rho N$ is chosen. 
In round $t$, the server samples $m$ participants $S_t$ uniformly without replacement, and the fraction of adversarial clients among the selected set is
$\beta_t = |S_t \cap A|/m$ with $\mathbb{E}[\beta_t] = \rho$ and $\mathbb{E}[\beta_t^2] = \rho^2 + \rho(1-\rho)\tfrac{N-m}{m(N-1)}$.
Each adversarial client may bias its gradient by at most
$\|\nabla \varphi_n^{\mathrm{adv}}(\theta) - \nabla \varphi_n(\theta)\|_2 \le \Gamma$, where $\Gamma = C_{\varepsilon}\varepsilon$ for a problem-dependent constant $C_{\varepsilon} > 0$ that scales the $\ell_\infty$ perturbation budget $\varepsilon = \|\boldsymbol{\delta}\|_\infty$ to the gradient bias magnitude.

The gradient bias bound $\Gamma$ depends on the attack strategy, perturbation budget, and model smoothness. For gradient-based poisoning attacks with $\ell_\infty$ perturbation budget $\varepsilon$, we can bound the gradient bias using the chain rule and smoothness of the loss. Specifically, if the loss $\ell(\cdot, y; \theta)$ is $L_{\ell}$-Lipschitz continuous in its first argument, then for a perturbed input $\tilde{\mathbf{X}} = \mathbf{X} + \boldsymbol{\delta}$ with $\|\boldsymbol{\delta}\|_\infty \le \varepsilon$, the gradient bias satisfies
\begin{equation}
\begin{aligned}
\|\nabla_\theta \ell(\tilde{\mathbf{X}}, y; \theta) - \nabla_\theta \ell(\mathbf{X}, y; \theta)\|_2 
&\le L_\ell \cdot \|\nabla_{\mathbf{X}} \nabla_\theta \ell\|_2 \cdot \|\boldsymbol{\delta}\|_2\\
&\le L_\ell \cdot \|\nabla_{\mathbf{X}} \nabla_\theta \ell\|_2 \cdot \sqrt{d} \cdot \varepsilon,
\end{aligned}
\end{equation}
where $d$ is the dimensionality of the input and the inequality uses $\|\boldsymbol{\delta}\|_2 \le \sqrt{d} \varepsilon$. In our analysis, we define $\Gamma$ as an upper bound on this quantity, which depends on both $\varepsilon$ and the model's Lipschitz properties.

\noindent\textbf{Reserve set.}
Server-side reserve SGD is unbiased with gradient variance $\sigma_r^2$ and small mismatch $\|\nabla \varphi_r(\theta)-\nabla \varphi(\theta)\|\le \varepsilon_r$. 

\medskip
\noindent We use the following algebraic decomposition of the local gradient.

\begin{lemma}[One-round descent]\label{lem:descent}
Under Assumptions~\ref{ass:smooth}--\ref{ass:drift}, after one training round with aggregation step size $\gamma_g \le 1/L$ and $r$ reserve-set SGD steps of size $\gamma_r \le 1/L$, the expected optimality gap contracts as
\begin{equation}
\begin{aligned}
\mathbb{E}[\varphi(\theta^{(t+1)}) - \varphi^\star] &\le (1-\mu\gamma_g)(1-\mu\gamma_r)^r\,\mathbb{E}[\varphi(\theta^{(t)}) - \varphi^\star] \\ &+C_g + C_r,
\label{eq:descent}
\end{aligned}
\end{equation}
where $C_g$ and $C_r$ are constants depending on $\sigma_g^2, \zeta^2, \Gamma^2$, and $\sigma_r^2$, with optimal minimizer value $\varphi^\star=\min_\theta \varphi(\theta)$.
\end{lemma}
\begin{proof}
See Appendix~\ref{appendix:proofs}.
\end{proof}

\begin{figure*}[!t]
  \centering
  \subfloat[\textnormal{\textit{Clean IID}}]{%
    \includegraphics[width=0.49\textwidth]{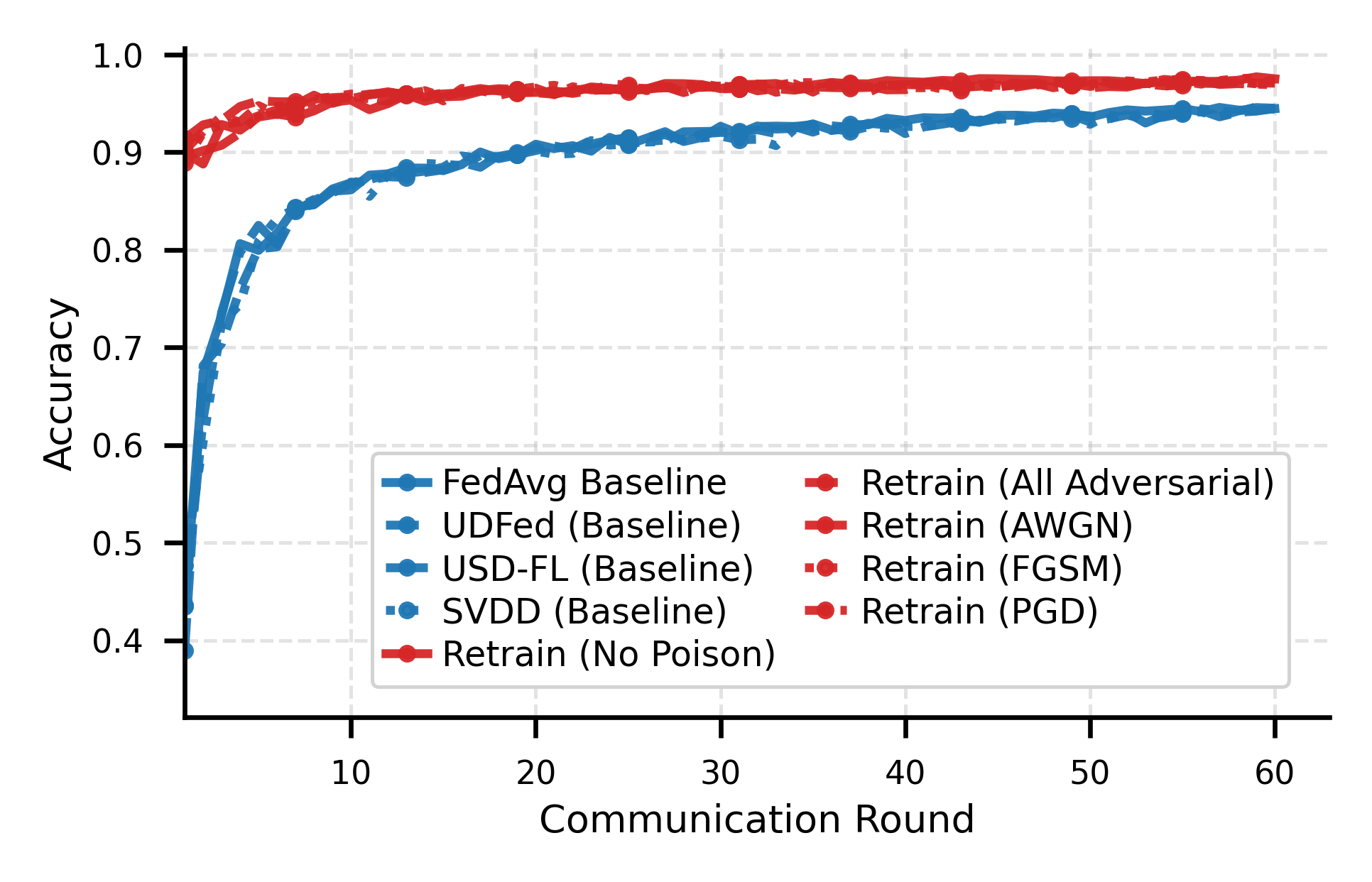}%
    \label{fig:iid_no_poison}}
  \hfill
  \subfloat[\textnormal{\textit{PGD IID}}]{%
    \includegraphics[width=0.49\textwidth]{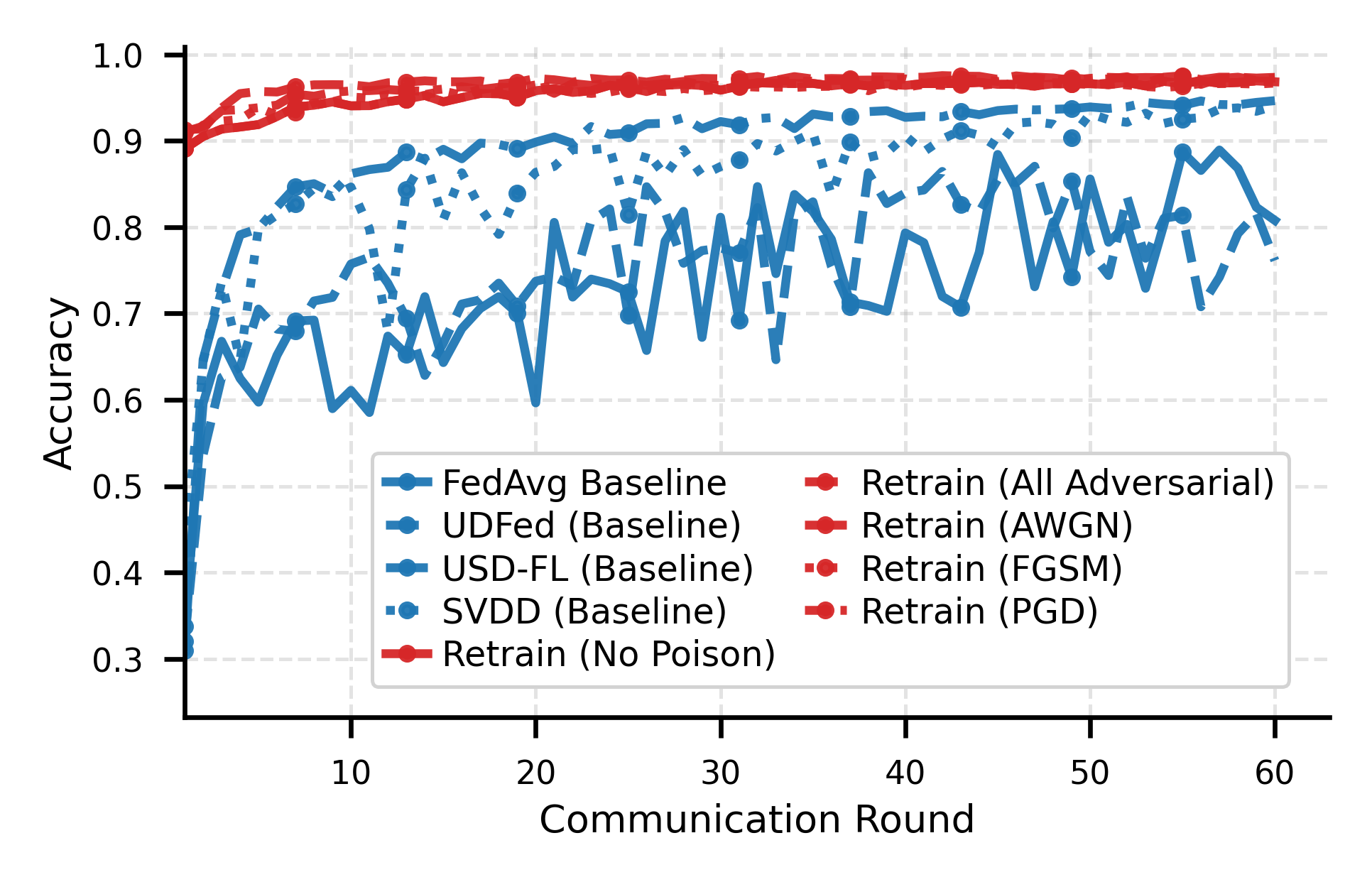}%
    \label{fig:iid_pgd}}

  \vspace{0.6em}

  \subfloat[\textnormal{\textit{AWGN IID}}]{%
    \includegraphics[width=0.49\textwidth]{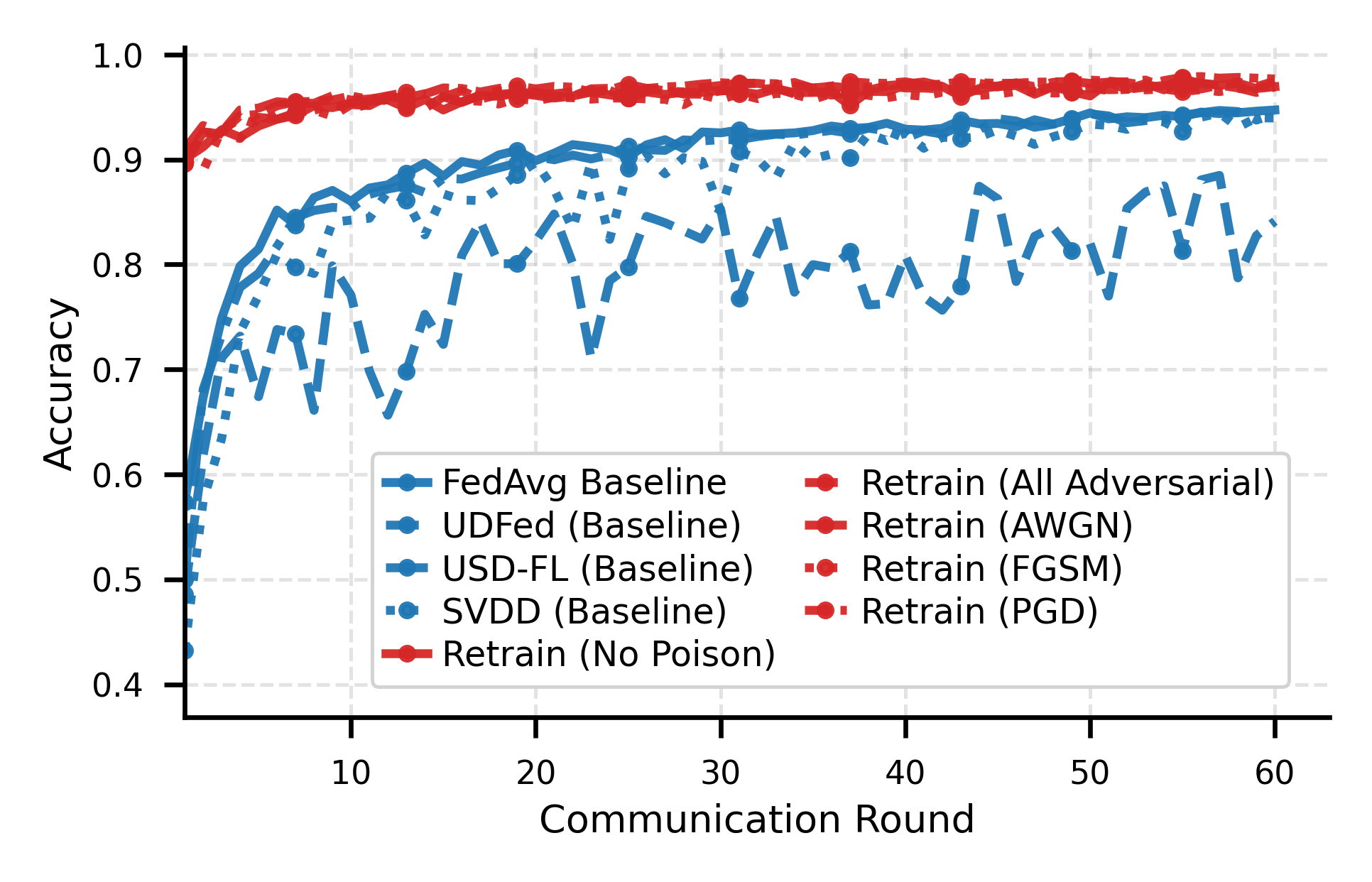}%
    \label{fig:iid_awgn}}
  \hfill
  \subfloat[\textnormal{\textit{FGSM IID}}]{%
    \includegraphics[width=0.49\textwidth]{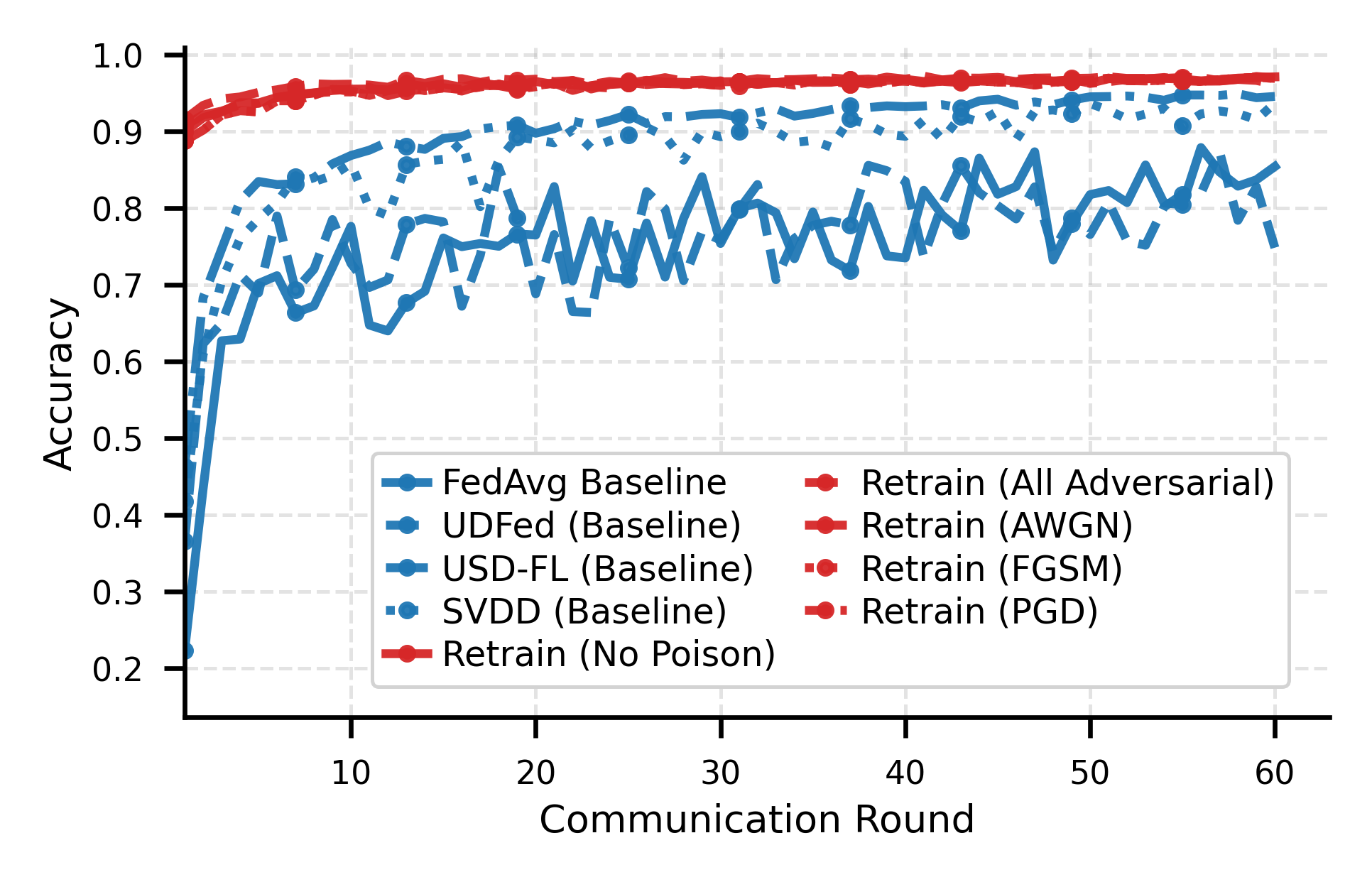}%
    \label{fig:iid_fgsm}}
  \caption{Global accuracy over communication rounds on AudioMNIST dataset under (a) clean, (b) PGD, (c) AWGN, and (d) FGSM poisoning attacks with IID data partition, comparing baselines with REVERB-FL framework methods (Retrain).}
  \label{fig:amnist_iid_2x2}
\end{figure*}

\begin{theorem}[Round-wise contraction with reserve retraining] \label{thm:conv-main}
Let the global objective $\varphi$ be $L$-smooth and $\mu$-strongly convex (Assumptions~\ref{ass:smooth}--\ref{ass:convex}).
In each communication round $t$, FedAvg is performed with effective step
$\gamma_g \!\le\! 1/L$, followed by $r$ unbiased reserve-set SGD steps
of size $\gamma_r \!\le\! 1/L$. Then
\begin{equation}
\mathbb{E}\!\left[\varphi(\theta^{(t+1)})-\varphi^\star\right]
\;\le\;
q\,\mathbb{E}\!\left[\varphi(\theta^{(t)})-\varphi^\star\right]
\;+\;
C',
\label{eq:main_bound}
\end{equation}
where the contraction factor and residual constant are
\begin{align}
q &= (1-\mu\gamma_g)\,(1-\mu\gamma_r)^{r}, \nonumber\\[3pt]
C' &= (1-\mu\gamma_r)^{r}\,
c_g(\gamma_g)\left(
\tfrac{c_s}{m}\sigma_g^2
+ c_\tau\,\zeta^2
+ \mathbb{E}[\beta_t^2]\,\Gamma^2
\right)\nonumber\\[3pt] &+ \tfrac{L\,\gamma_r^{2}\,r}{2}\,\sigma_r^2
\label{eq:main_constants}
\end{align}
with $c_g(\gamma_g)=\tfrac{\gamma_g}{2a}+\tfrac{L\gamma_g^2}{2}$ and $c_\tau=\tfrac{\tau(\tau-1)}{2}\eta^2 L^2$ for any $a \in (0,1)$.
\end{theorem}

\begin{proof}
Applying strong convexity (Assumption~\ref{ass:convex}) to Lemma~\ref{lem:descent} yields the contraction rate $q = (1-\mu\gamma_g)(1-\mu\gamma_r)^r < 1$. Unrolling the recursion and bounding the geometric series gives~\eqref{eq:main_bound} with steady-state constant~\eqref{eq:main_constants}. See Appendix~\ref{appendix:proofs} for full proof.
\end{proof}

Setting $r\!=\!0$ recovers the baseline FedAvg rate with $q_{\mathrm{FA}}=1-\mu\gamma_g$. Since $(1-\mu\gamma_r)^{r}\!<\!1$ for any $r\!\ge\!1$, the proposed reserve retraining yields faster contraction and a smaller steady-state error compared to baseline FedAvg. The constant $C'$ combines stochastic variance (scaled by $c_s/m$), client heterogeneity ($c_\tau\zeta^2$), and adversarial bias $\Gamma^2 = (C_{\varepsilon}\varepsilon)^2$ scaled by $\mathbb{E}[\beta_t^2]$ (where $\varepsilon = \|\boldsymbol{\delta}\|_\infty$ is the perturbation budget); reserve updates dampen these through $(1-\mu\,\gamma_r)^r$ while adding a small $\sigma_r^2$ term. 

\section{Performance Evaluation} \label{results}
In this section, we first describe our experimental setup and FL architecture (Sec.~\ref{sec:exp_setup}). Next, we evaluate the efficacy of our framework under model poisoning attacks on IID data partitions (Sec.~\ref{iid_res}) and non-IID data partitions (Sec.~\ref{non-res}). We compare REVERB-FL variants against multiple baselines, reporting global model accuracy on clean test inputs and under each considered poisoning attacks. 

\begin{figure*}[!t]
  \centering
  \subfloat[\textnormal{\textit{Clean IID}}]{%
    \includegraphics[width=0.49\textwidth]{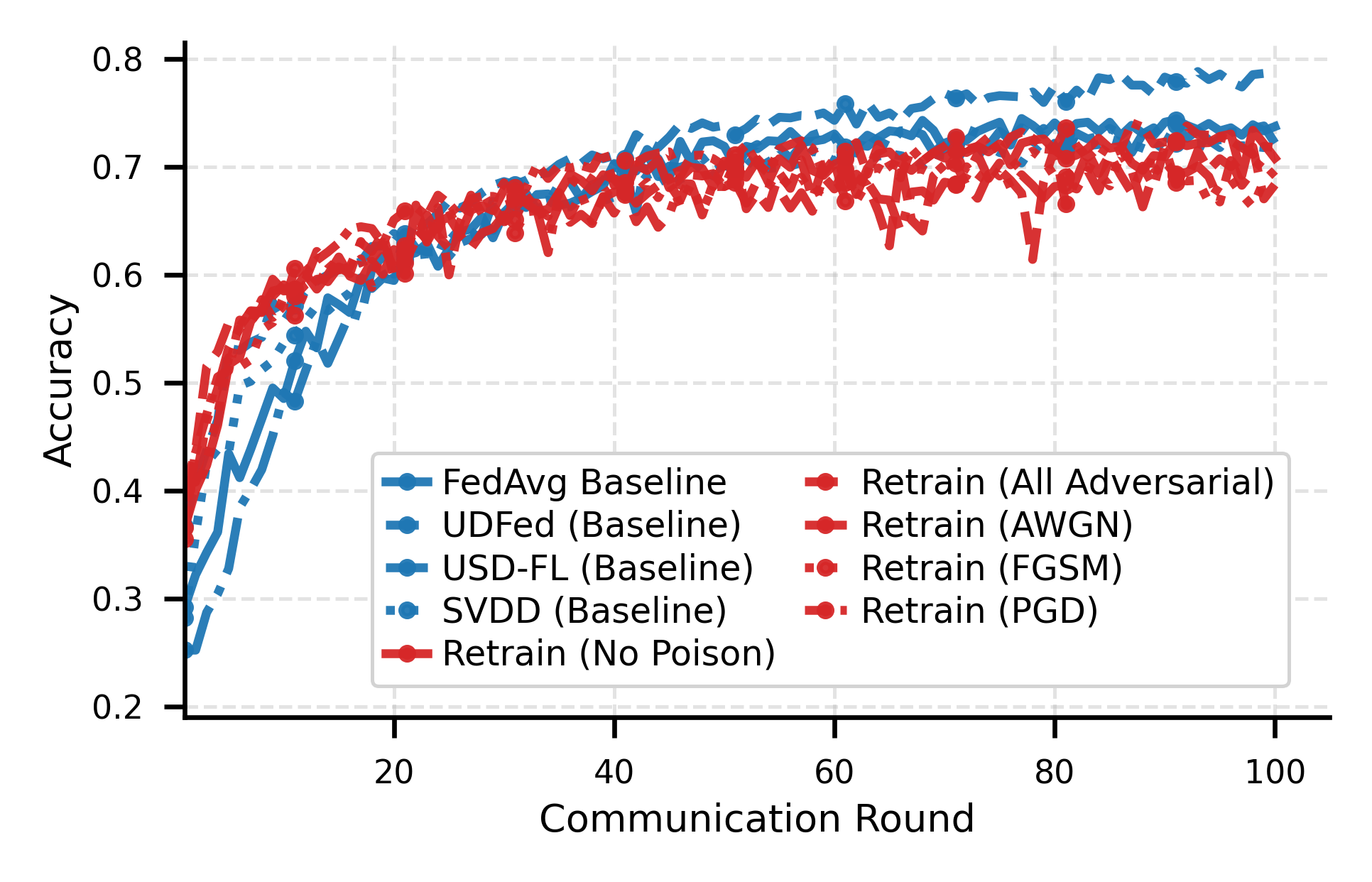}%
    \label{fig:us_iid_no_poison}}
  \hfill
  \subfloat[\textnormal{\textit{PGD IID}}]{%
    \includegraphics[width=0.49\textwidth]{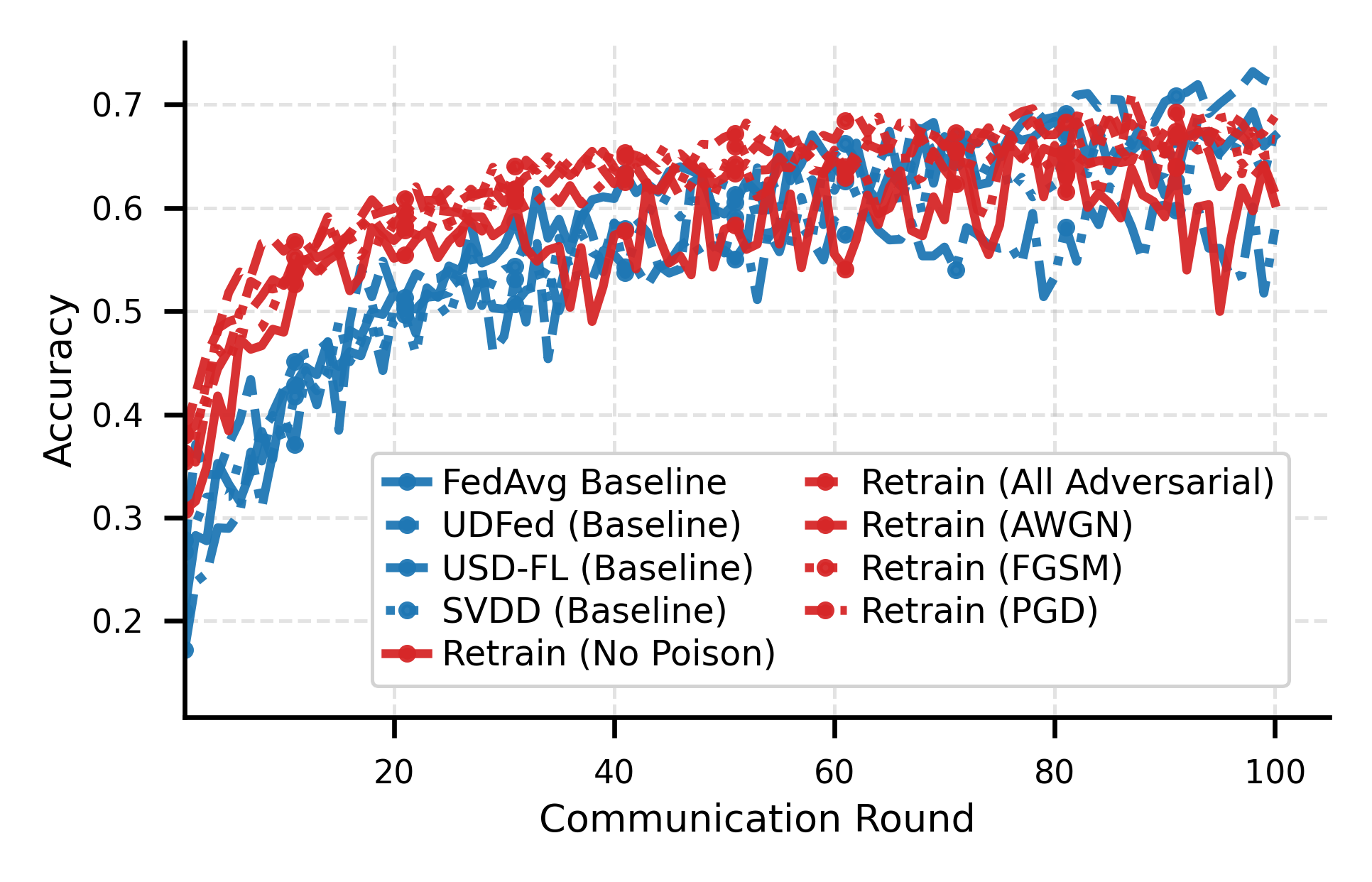}%
    \label{fig:us_iid_pgd}}

  \vspace{0.6em}

  \subfloat[\textnormal{\textit{AWGN IID}}]{%
    \includegraphics[width=0.49\textwidth]{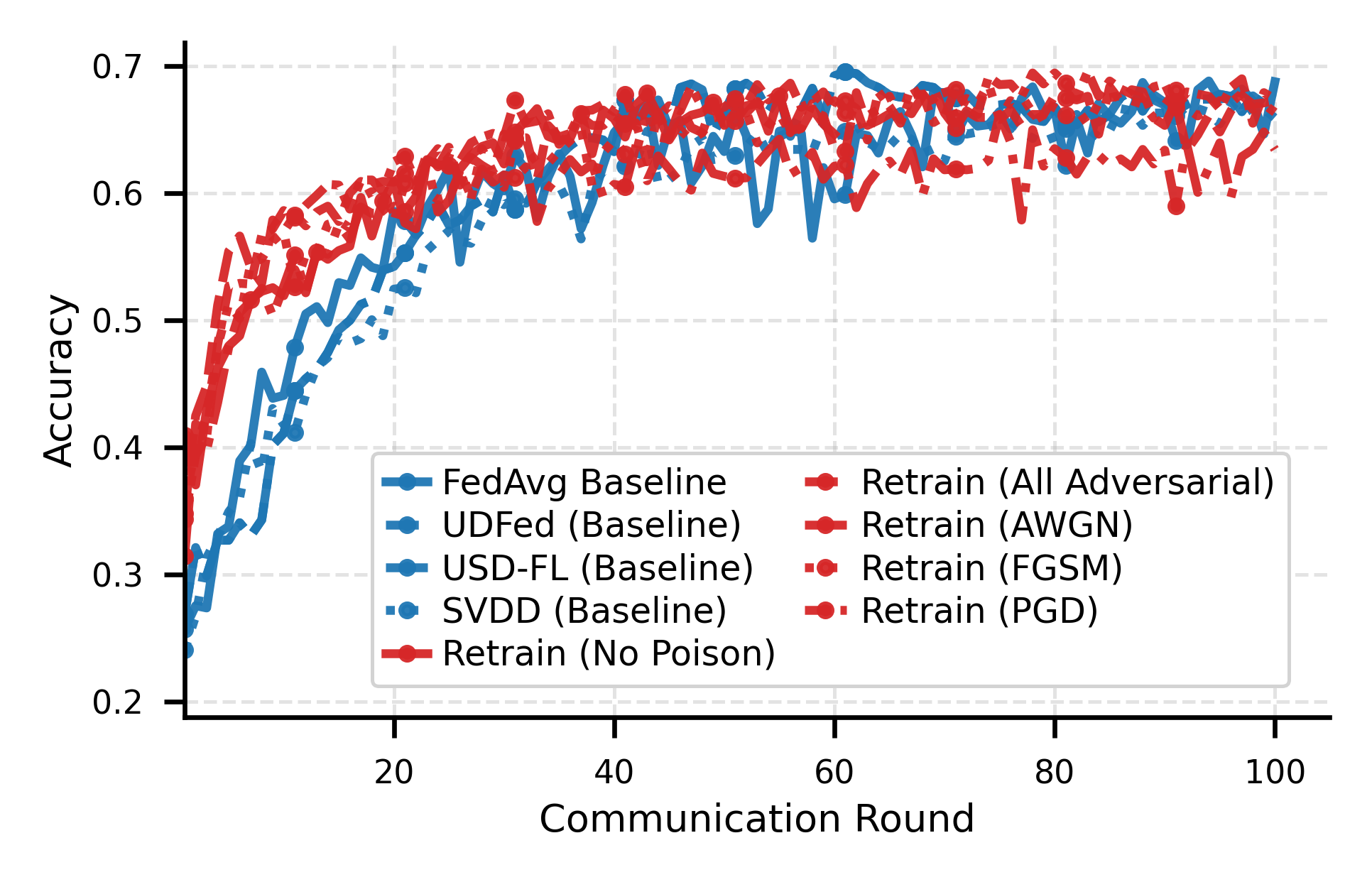}%
    \label{fig:us_iid_awgn}}
  \hfill
  \subfloat[\textnormal{\textit{FGSM IID}}]{%
    \includegraphics[width=0.49\textwidth]{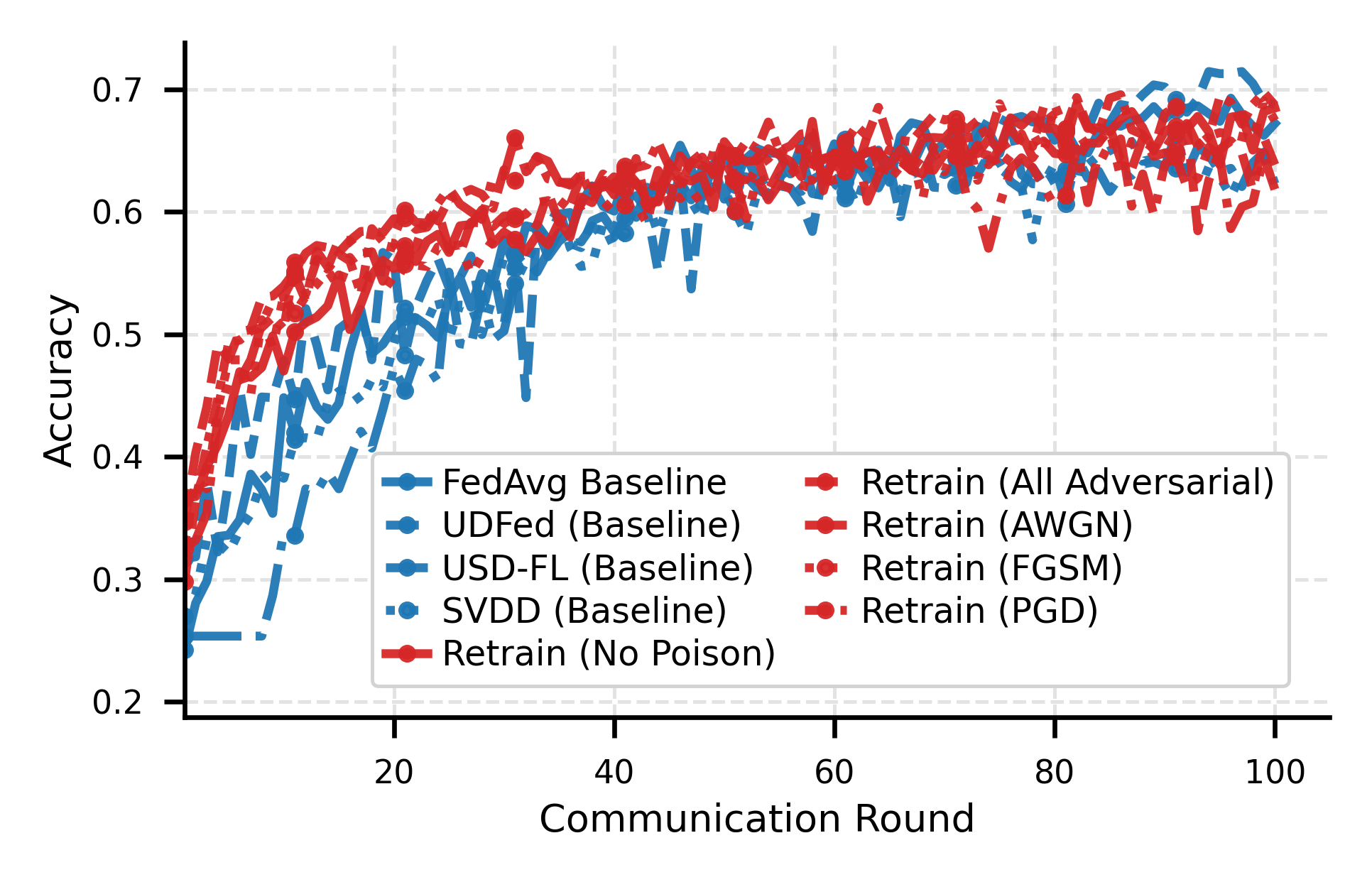}%
    \label{fig:us_iid_fgsm}}
  \caption{Global accuracy over communication rounds on UrbanSound8K dataset under (a) clean, (b) PGD, (c) AWGN, and (d) FGSM poisoning attacks with IID data partition, comparing baselines with REVERB-FL framework methods (Retrain).}
  \label{fig:us8k_iid_2x2}
\end{figure*}

\subsection{Experimental Setup} \label{sec:exp_setup}
We evaluate REVERB-FL under a \emph{training-time data-poisoning} setting, where, in each communication round, a fixed subset of selected clients perturb their local \emph{training} data to create poisoned data. Poisoned clients keep labels unchanged while the server and clean clients are honest. Global evaluation uses the test set.

\textit{Poisoning mechanisms.}
We instantiate three perturbation families applied to client-side training inputs:
(i) \textbf{FGSM-poison} with $\ell_\infty$ budget $\varepsilon = 0.02$, crafting adversarial examples against the current global model;
(ii) \textbf{PGD-poison} with $\varepsilon = 0.02$, step size $\varepsilon/50$, and $50$ iterations;
(iii) \textbf{AWGN-poison} by adding $\boldsymbol{n} \sim \mathcal{N}(0,\sigma^2 I)$ with $\sigma = 0.03$ and clipping to $[0,1]$.
We designate a fixed adversarial fraction $\rho = 0.5$ (i.e., 50\% of clients apply poisoning) across all experiments. Attack parameters were selected via sensitivity analysis on centralized models, where these values produce significant accuracy degradation (e.g., FGSM reduces accuracy to ~20-30\% at $\varepsilon = 0.02$) while maintaining realistic perturbation budgets for audio spectrograms.
We fix the poisoned-client ratio per round and report it alongside results.

\textit{Defense (server-side).}
REVERB-FL defends via a 5\% \emph{server reserve set}, obtained by stratified sampling from the clients. The sampled data is transmitted to the server and removed from client datasets, ensuring class balance and disjointness from all local training data. The reserve set is used for: (a) pretraining the global model for 3 epochs before round~1; and (b) one epoch of server-side retraining after each aggregation, performing $r$ SGD steps corresponding to one epoch over $\mathcal{D}_r$ with batch size $B_r = 32$ (compared to client local SGD with $\tau$ steps and batch size $B = 16$). Reserve retraining uses either clean reserve batches (\textbf{Retrain (No Poison)}) or adversarially augmented reserve batches generated with FGSM, PGD, AWGN, or a mixture (\textbf{Retrain (FGSM)}/\textbf{Retrain (PGD)}/\textbf{Retrain (AWGN)}/\textbf{Retrain (All Adversarial)}). The full \textbf{REVERB-FL} configuration combines Retrain with \textbf{All Adversarial}. The server-side retraining overhead is lightweight and predictable: $\mathcal{D}_r$ is collected once before training begins and remains fixed throughout, so the number of reserve SGD steps per round $r = \lceil|\mathcal{D}_r|/B_r\rceil$ is constant across all communication rounds. The retraining cost is therefore determined entirely at initialization and does not grow during training, making it straightforward to budget in practice. For very large federations, the absolute size of $\mathcal{D}_r$ can be capped at a fixed budget independent of total dataset size, trading a slightly larger $\varepsilon_r$ for reduced server overhead.

\begin{figure*}[!t]
  \centering
  \subfloat[\textnormal{\textit{Clean non-IID}}]{%
    \includegraphics[width=0.49\textwidth]{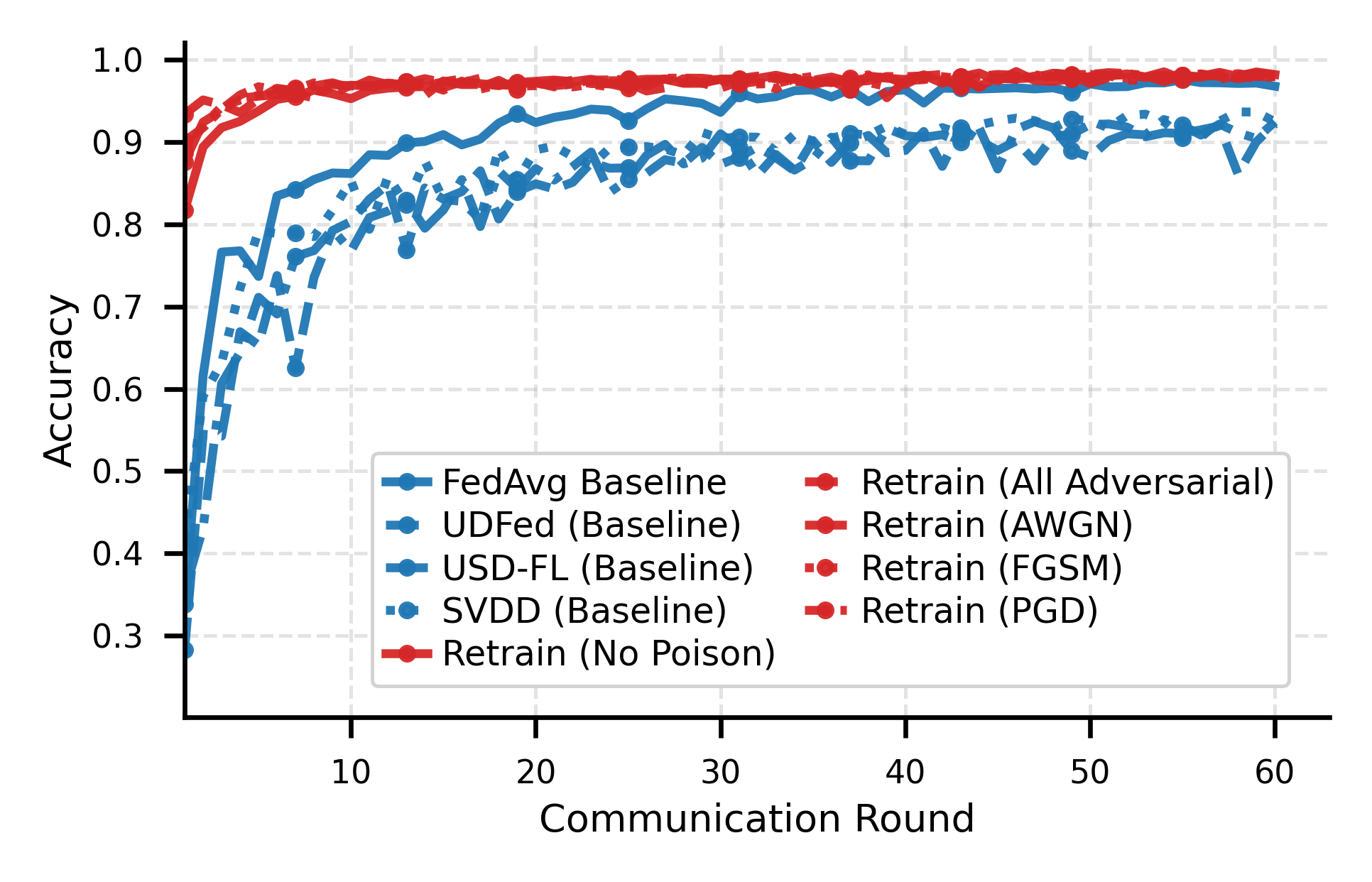}%
    \label{fig:non_iid_no_poison}}
  \hfill
  \subfloat[\textnormal{\textit{PGD non-IID}}]{%
    \includegraphics[width=0.49\textwidth]{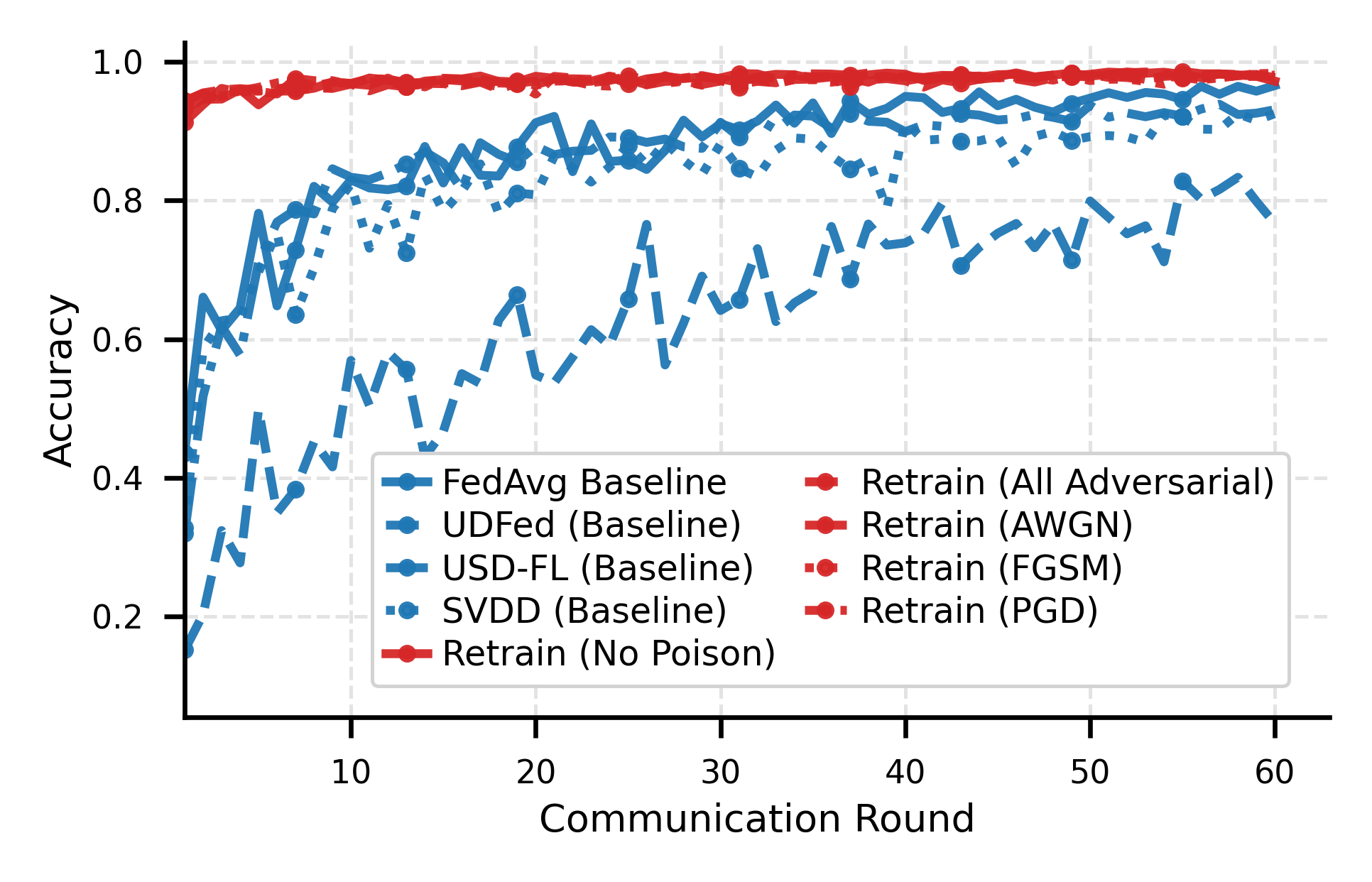}%
    \label{fig:non_iid_pgd}}

  \vspace{0.6em}

  \subfloat[\textnormal{\textit{AWGN non-IID}}]{%
    \includegraphics[width=0.49\textwidth]{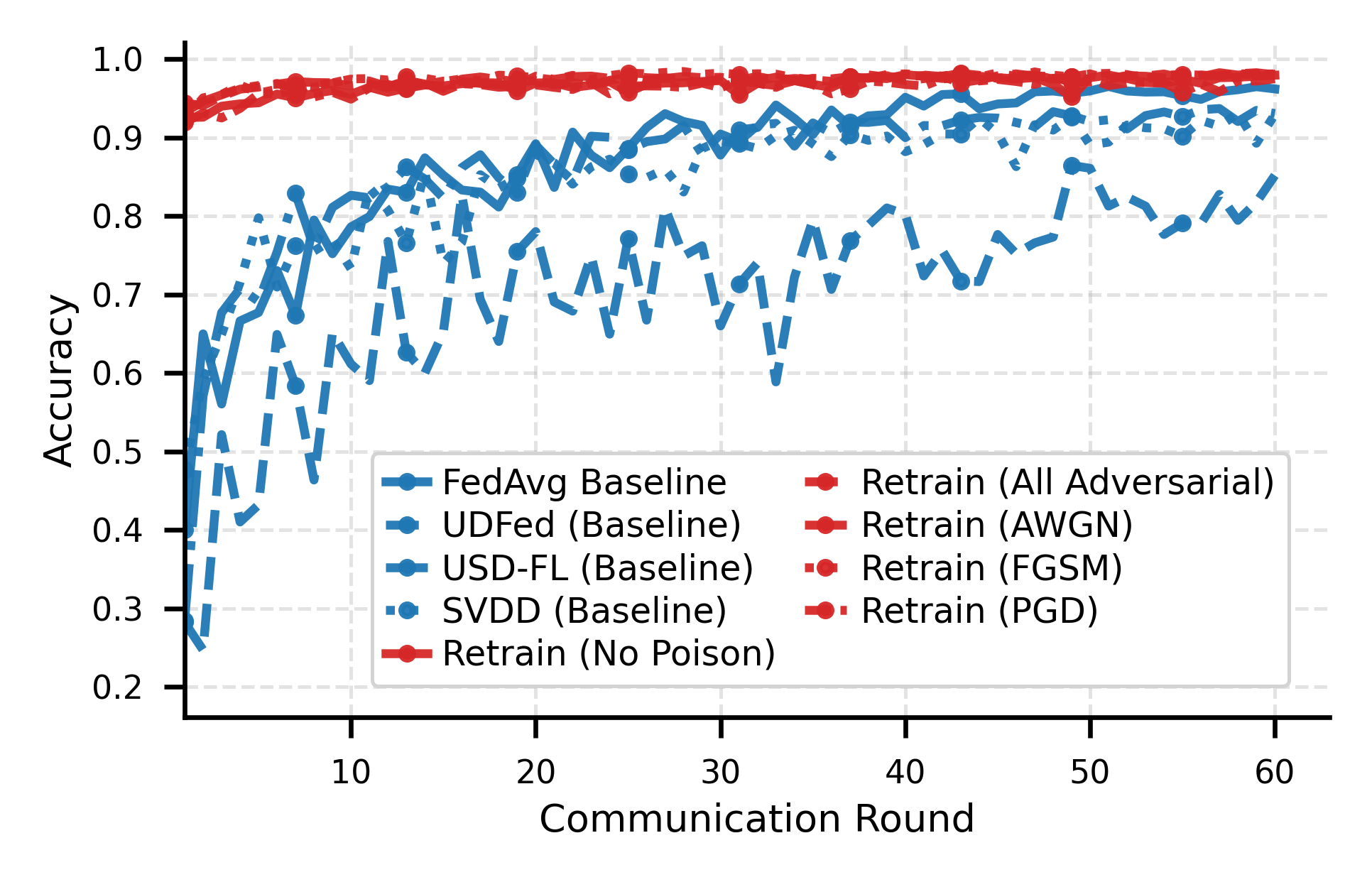}%
    \label{fig:non_iid_awgn}}
  \hfill
  \subfloat[\textnormal{\textit{FGSM non-IID}}]{%
    \includegraphics[width=0.49\textwidth]{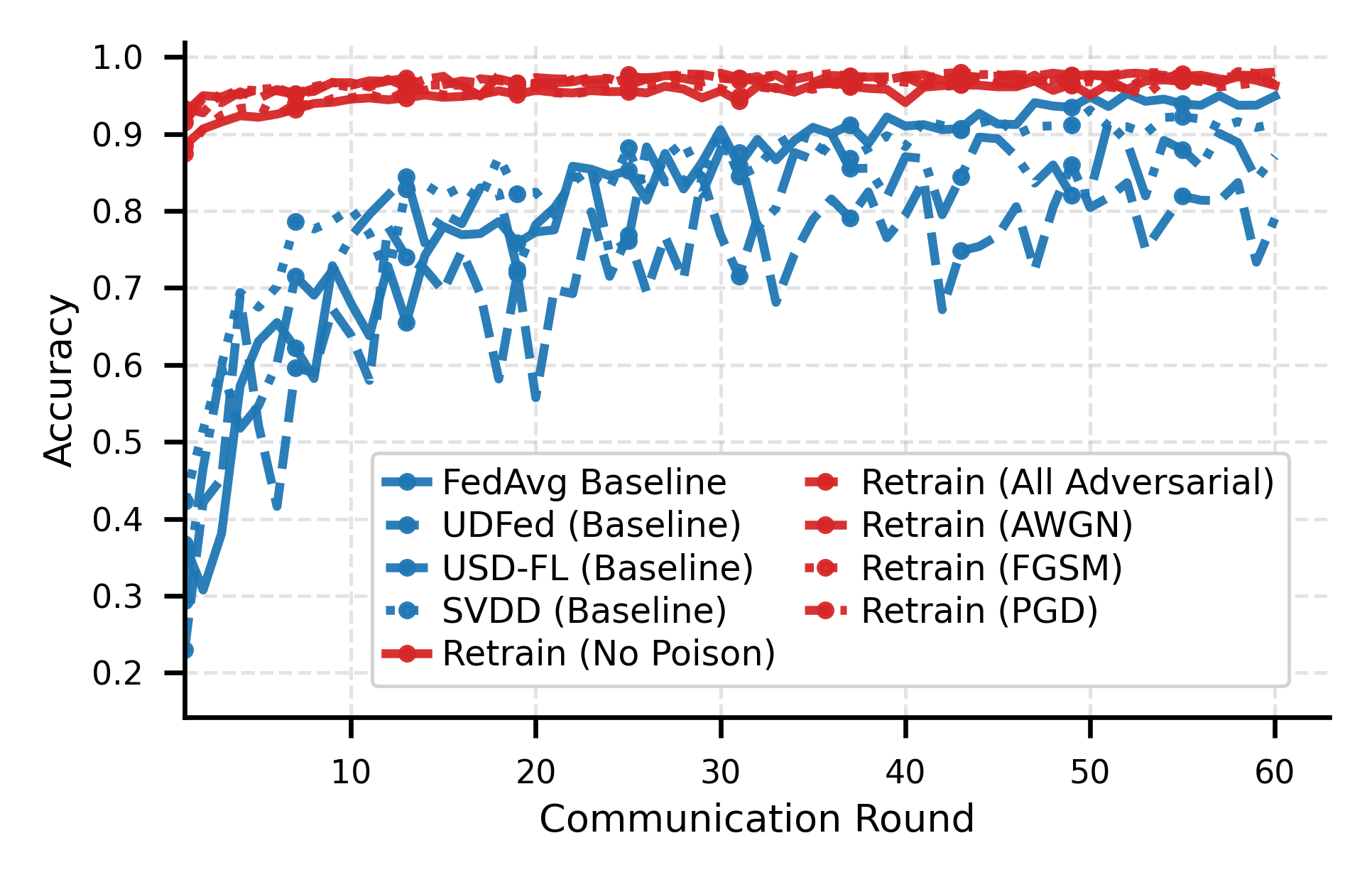}%
    \label{fig:non_iid_fgsm}}
  \caption{Global accuracy over communication rounds on AudioMNIST dataset under (a) clean, (b) PGD, (c) AWGN, and (d) FGSM poisoning attacks with non-IID data partition, comparing baselines with REVERB-FL framework methods (Retrain).}
  \label{fig:amnist_non_iid_2x2}
\end{figure*}

\textit{Datasets, partitioning, and preprocessing.}
We use the \emph{AudioMNIST} \cite{becker2023audiomnistexploringexplainableartificial} and \emph{UrbanSound8K} \cite{urbansound} datasets. Audio is resampled to 16\,kHz and transformed into complex STFT spectrograms using Hann windows with window length $L_w = 1024$ samples, hop size $h = 512$ samples, and FFT size $F = 1024$, yielding frequency bins $n_f = F/2 + 1 = 513$ and time frames $T$ dependent on utterance length. The complex spectrogram is split into real and imaginary components to form the input tensor $\mathbf{X} \in \mathbb{R}^{513 \times T \times 2}$ \cite{stft}. All input tensors are normalized per utterance to zero mean and unit variance, then clipped element-wise to the admissible set $\mathcal{X} = [-3, 3]^{n_f \times T \times 2}$ to ensure bounded input values. The 5\% reserve set is stratified by class and disjoint from client data. IID partitions use equal random splits; non-IID partitions use Dirichlet label-skew with concentration $\alpha = 0.5$. In this setting, for each class, sample indices are first grouped by label and then divided among clients according to proportions drawn from $\text{Dir}(\alpha)$. These proportions determine how many samples of each class each client receives, producing overlapping but imbalanced label distributions across clients. With $\alpha=0.5$, some clients concentrate on a few dominant classes while others retain more mixed distributions, yielding moderate heterogeneity representative of real-world non-IID audio data. 

\textit{Model architecture and training.}
The classifier is a spectrogram CNN \cite{spectrograms} with three convolutional blocks (32, 64, 128 filters with $3{\times}3$ kernels, batch normalization, ReLU activation, and $2{\times}2$ max-pooling), followed by a 128-unit dense layer with dropout (0.5) and softmax ($\sim\!1.1$M parameters). Our complete CNN architecture is shown in Table \ref{tab:cnn_architecture}. Clients use Adam optimizer with exponential learning rate decay (initial learning rate $\eta = 1 \times 10^{-4}$, decay rate 0.9, decay steps 1000), batch size $B = 16$, $L_2$ weight decay coefficient $\lambda = 1 \times 10^{-4}$, and dropout probability $p = 0.5$ applied to the penultimate layer. Reserve retraining at the server uses identical Adam settings but with larger batch size $B_r = 32$ to stabilize updates. These hyperparameters were selected via grid search on AudioMNIST validation data.

\begin{table}[t]
\centering
\caption{CNN Architecture for Audio Classification}
\label{tab:cnn_architecture}
\begin{tabular}{lcc}
\toprule
\textbf{Layer} & \textbf{Activation} & \textbf{Output Shape} \\
\midrule
Input & -- & $(n_f \times T \times 2)$ \\
\midrule
Conv2D (32 filters, $3\times3$) & ReLU & $(n_f \times T \times 32)$ \\
MaxPool2D ($2\times2$) & -- & $(n_f/2 \times T/2 \times 32)$ \\
\midrule
Conv2D (64 filters, $3\times3$) & ReLU & $(n_f/2 \times T/2 \times 64)$ \\
MaxPool2D ($2\times2$) & -- & $(n_f/4 \times T/4 \times 64)$ \\
\midrule
Conv2D (128 filters, $3\times3$) & ReLU & $(n_f/4 \times T/4 \times 128)$ \\
MaxPool2D ($2\times2$) & -- & $(n_f/8 \times T/8 \times 128)$ \\
\midrule
Flatten & -- & $(n_f \cdot T \cdot 128 / 64)$ \\
Dense (128 units) & ReLU & $128$ \\
Dropout ($p=0.5$) & -- & $128$ \\
Dense ($K$ classes) & Softmax & $K$ \\
\midrule
\multicolumn{3}{l}{\small Total parameters: $\sim$1.1M} \\
\end{tabular}
\end{table}

\textit{Federated learning protocol.}
Before federated training, the global model is pretrained on the reserve set $\mathcal{D}_r$ for 3 epochs. We adopt the FedAvg aggregation rule \cite{fedAvg_baseline} with sampling fraction $C=0.6$. For AudioMNIST, we use $N=10$ clients performing $\tau=10$ local SGD steps per round with batch size $B=16$, trained for $R=60$ communication rounds. For UrbanSound8K, we use $N=8$ clients performing $\tau=30$ local SGD steps per round with batch size $B=16$, trained for $R=100$ communication rounds. After each aggregation, the server performs reserve-set retraining for one epoch through $\mathcal{D}_r$ (approximately $r = \lceil|\mathcal{D}_r|/B_r\rceil$ SGD steps) using batch size $B_r=32$.

All REVERB-FL configurations use the FedAvg aggregation rule (Eq.~\eqref{eq:fedavg}) and differ only in their server-side reserve-set retraining strategy. We evaluate the following configurations:
\begin{enumerate}
    \item \textbf{Retrain (No Poison)} (Reserve Set retraining): FedAvg augmented with server-side retraining on a clean 5\% reserve set before training and after each aggregation round, providing stabilization without adversarial augmentation.
    \item \textbf{Retrain (FGSM)}, \textbf{Retrain (PGD)}, \textbf{Retrain (AWGN)}: Reserve set retraining with single-attack adversarial augmentation, where the reserve set is augmented with adversarial examples generated using FGSM \cite{fgsm}, PGD \cite{pgd}, or AWGN, respectively, to provide attack-specific robustness.
    \item \textbf{Retrain (All Adversarial)}: Reserve set retraining with mixed adversarial augmentation, where the reserve set is augmented with adversarial examples from all three attack types (FGSM, PGD, AWGN) to provide robustness across multiple poisoning strategies.
\end{enumerate}

These configurations are compared against baseline FedAvg \cite{fedAvg_baseline} and state-of-the-art FL defense methods introduced in Sec.~\ref{iid_res}. We report clean accuracy and robust accuracy measured on clean test inputs \emph{after training under poisoning}, with per-attack results and their mean.

\begin{figure*}[!t]
  \centering
  \subfloat[\textnormal{\textit{Clean non-IID}}]{%
    \includegraphics[width=0.49\textwidth]{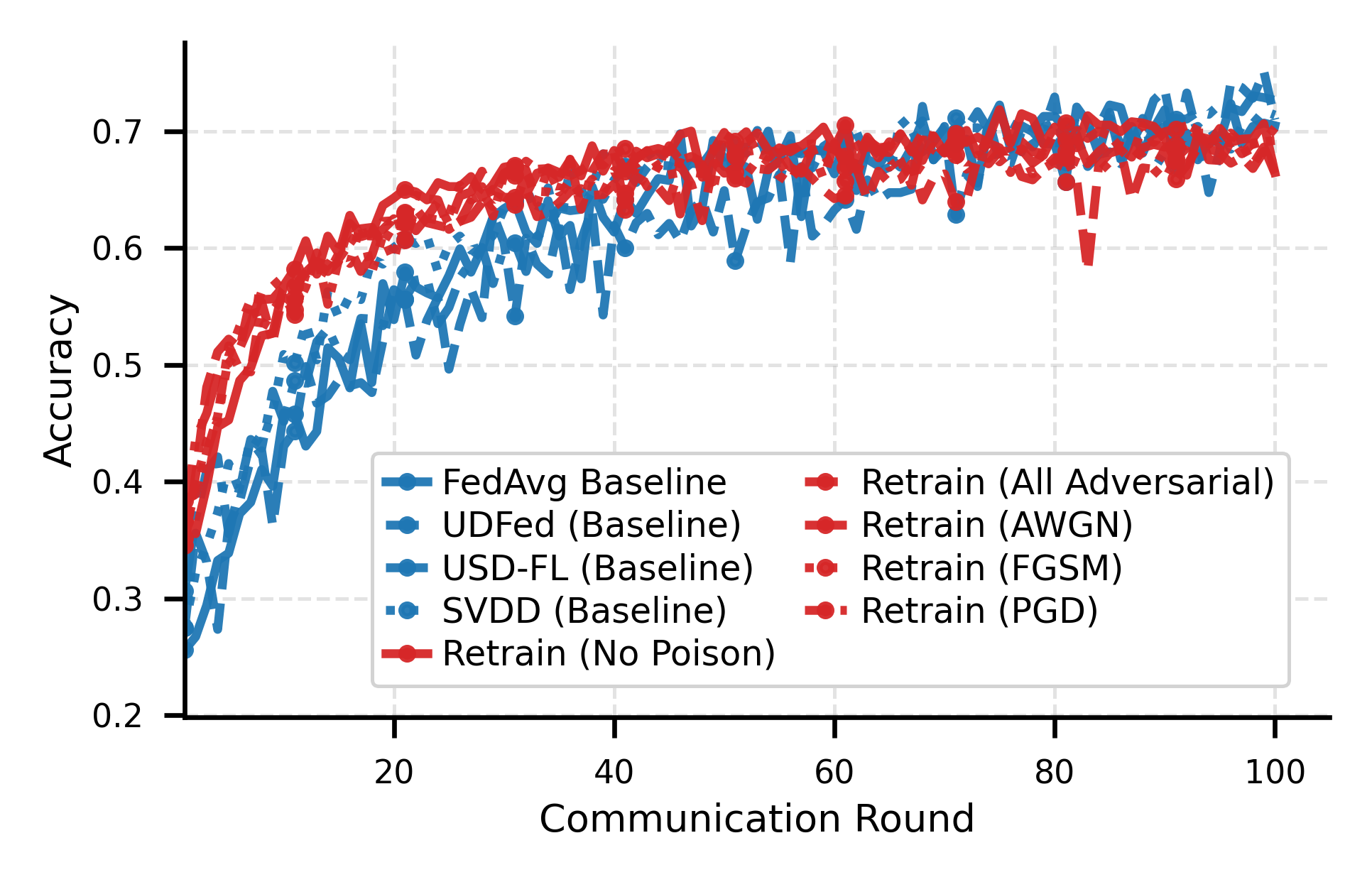}%
    \label{fig:us_non_iid_no_poison}}
  \hfill
  \subfloat[\textnormal{\textit{PGD non-IID}}]{%
    \includegraphics[width=0.49\textwidth]{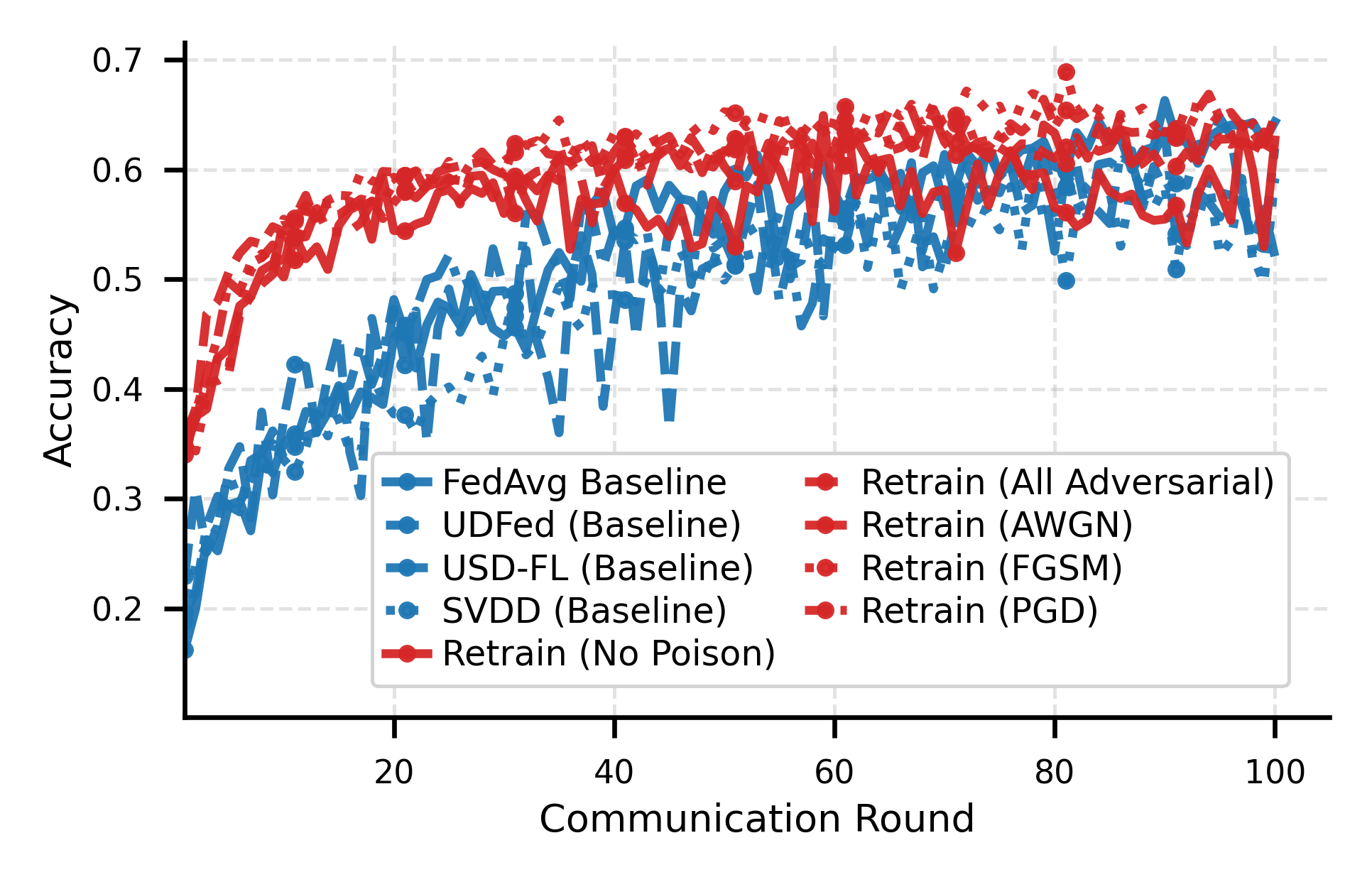}%
    \label{fig:us_non_iid_pgd}}

  \vspace{0.6em}

  \subfloat[\textnormal{\textit{AWGN non-IID}}]{%
    \includegraphics[width=0.49\textwidth]{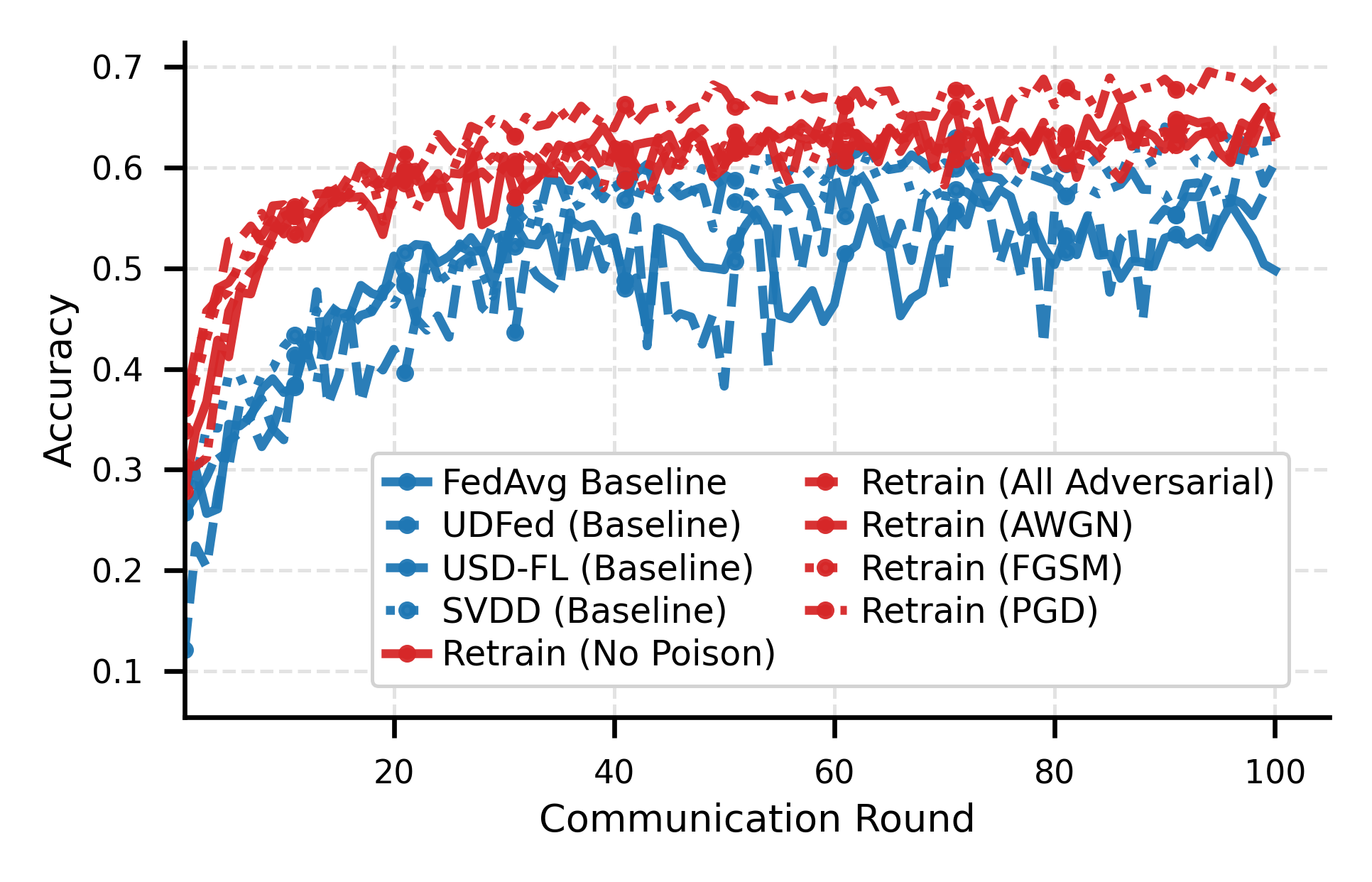}%
    \label{fig:us_non_iid_awgn}}
  \hfill
  \subfloat[\textnormal{\textit{FGSM non-IID}}]{%
    \includegraphics[width=0.49\textwidth]{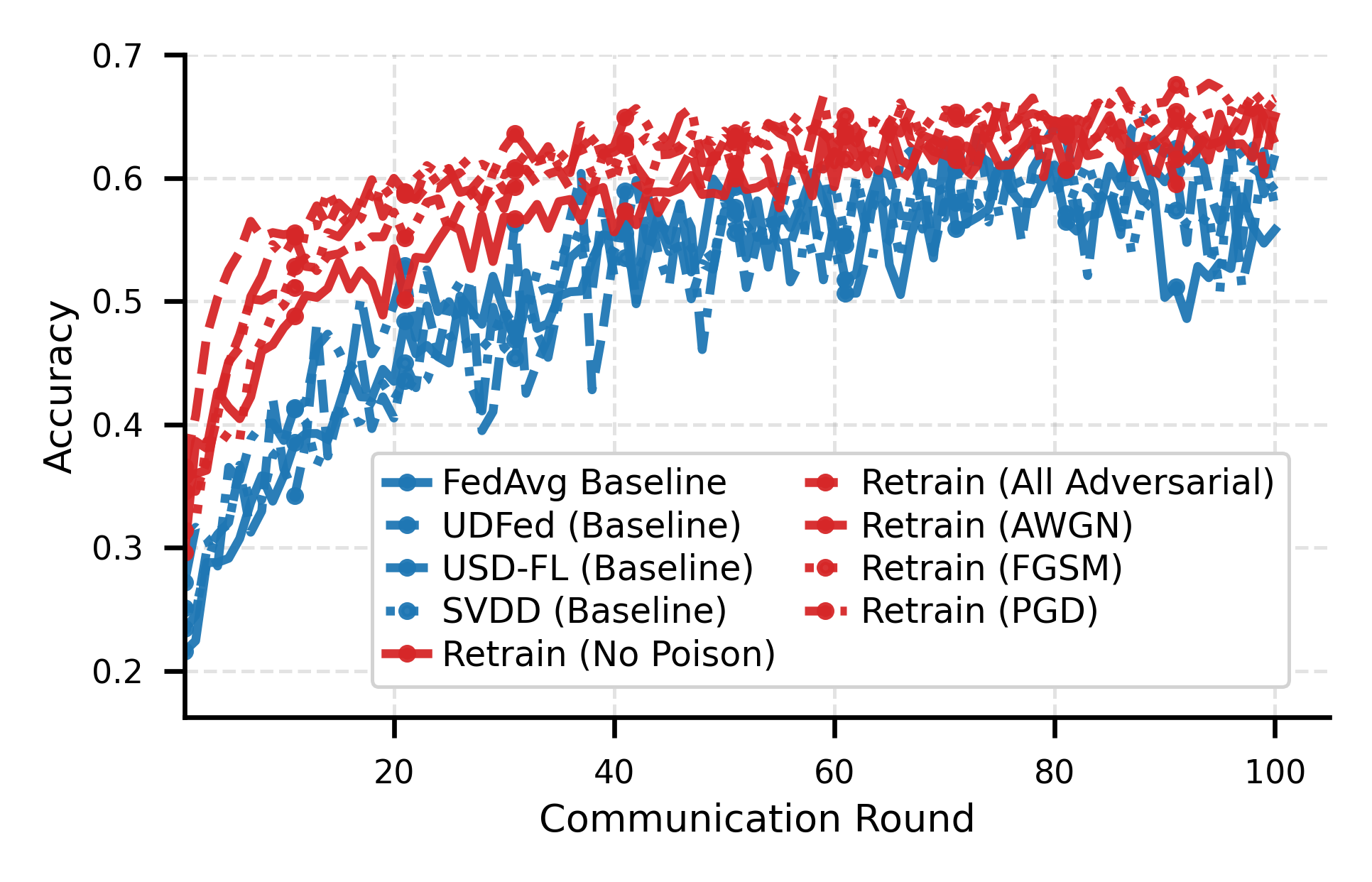}%
    \label{fig:us_non_iid_fgsm}}
  \caption{Global accuracy over communication rounds on UrbanSound8K dataset under (a) clean, (b) PGD, (c) AWGN, and (d) FGSM poisoning attacks with non-IID data partition, comparing baselines with REVERB-FL framework methods (Retrain).}
  \label{fig:us8k_non_iid_2x2}
\end{figure*}

\subsection{i.i.d. results}\label{iid_res}
\textit{Baseline comparison.} We compare REVERB-FL with baseline FedAvg \cite{fedAvg_baseline}, which performs standard weighted aggregation without reserve-set retraining, serving as the primary undefended reference. Additionally, we compare against three state-of-the-art FL defense methods selected to represent the major categories of existing defenses: (1) \textbf{USD-FL} \cite{evasion_attacks_fl}, which detects adversarial clients by analyzing logit distributions and computing pairwise 1-Wasserstein distances between client updates, using an adaptive threshold function without requiring knowledge of the number of adversaries. USD-FL is directly relevant as it was designed for signal classification in FL settings and represents the strongest existing client-detection defense for this domain. (2) \textbf{Deep SVDD} \cite{deep_svdd_cite}, which trains a deep one-class classifier using Deep Support Vector Data Description on benign model parameters from a root dataset, learning to detect anomalies by mapping parameters into a hypersphere of minimum volume while employing noise injection to prevent hypersphere collapse. Deep SVDD represents the class of anomaly-detection-based defenses \cite{krum, bulyan} that operate on model parameters rather than data, providing a strong representative of parameter-space filtering approaches. (3) \textbf{UDFed} \cite{udfed_cite}, which combines three defense strategies: anonymous obfuscation with differential privacy, joint similarity-based collusion detection using Kernel Density Estimation (T-KDE), and iterative low-rank approximation-based anomaly detection to amplify differences between benign and malicious gradients. UDFed represents a recent comprehensive multi-strategy defense combining differential privacy, collusion detection, and low-rank approximation \cite{udfed_cite}, reflecting the current state of the art in general FL defense. All baselines were configured using hyperparameters reported in their respective papers and evaluated under the same FL protocol, ensuring a fair comparison. Note that none of the baselines are specifically designed for the audio domain, which is intentional because it allows us to compare against the strongest general FL defenses to demonstrate that robust domain-agnostic approaches are insufficient and that domain-aware server-side retraining provides additional gains. Moreover, to the best of our knowledge, audio-specific defenses for FL have not been previously investigated.

On \textbf{AudioMNIST}, all methods perform comparably in the clean setting (Fig.~\ref{fig:iid_no_poison}), with reserve methods converging slightly faster to 97--98\%. Under gradient-based poisoning, significant gaps emerge: baselines degrade substantially under FGSM and PGD poisoning (Figs.~\ref{fig:iid_fgsm},~\ref{fig:iid_pgd}), while all reserve methods maintain above 96\%, with attack-matched configurations performing strongest. Table~\ref{tab:summary_results} summarizes the final-round accuracy across all experimental conditions, with the best-performing method per column highlighted in bold; subsequent references direct the reader to this table for detailed numerical comparisons. AWGN poisoning is less potent, with all methods remaining above 93\% (Fig.~\ref{fig:iid_awgn}). 

On \textbf{UrbanSound8K}, the 10-class task is more challenging overall (Figs.~\ref{fig:us_iid_no_poison}--\ref{fig:us_iid_fgsm}). Baselines degrade to 52--58\% under gradient-based attacks, with UDFed showing severe instability under AWGN poisoning. Reserve methods consistently sustain 60--68\%, with \emph{Retrain (PGD)} achieving the strongest results across most conditions (Table~\ref{tab:summary_results}).

Overall, IID experiments demonstrate that reserve-set retraining provides 5--15\% accuracy improvements over baselines under gradient-based attacks. Attack-specific retraining maximizes robustness against matched attacks, while \textbf{Retrain (All Adversarial)} provides consistent cross-attack robustness across all conditions.

\begin{figure*}[!t]
  \centering
  \subfloat[\textnormal{\textit{AudioMNIST IID}}]{%
    \includegraphics[width=0.49\textwidth]{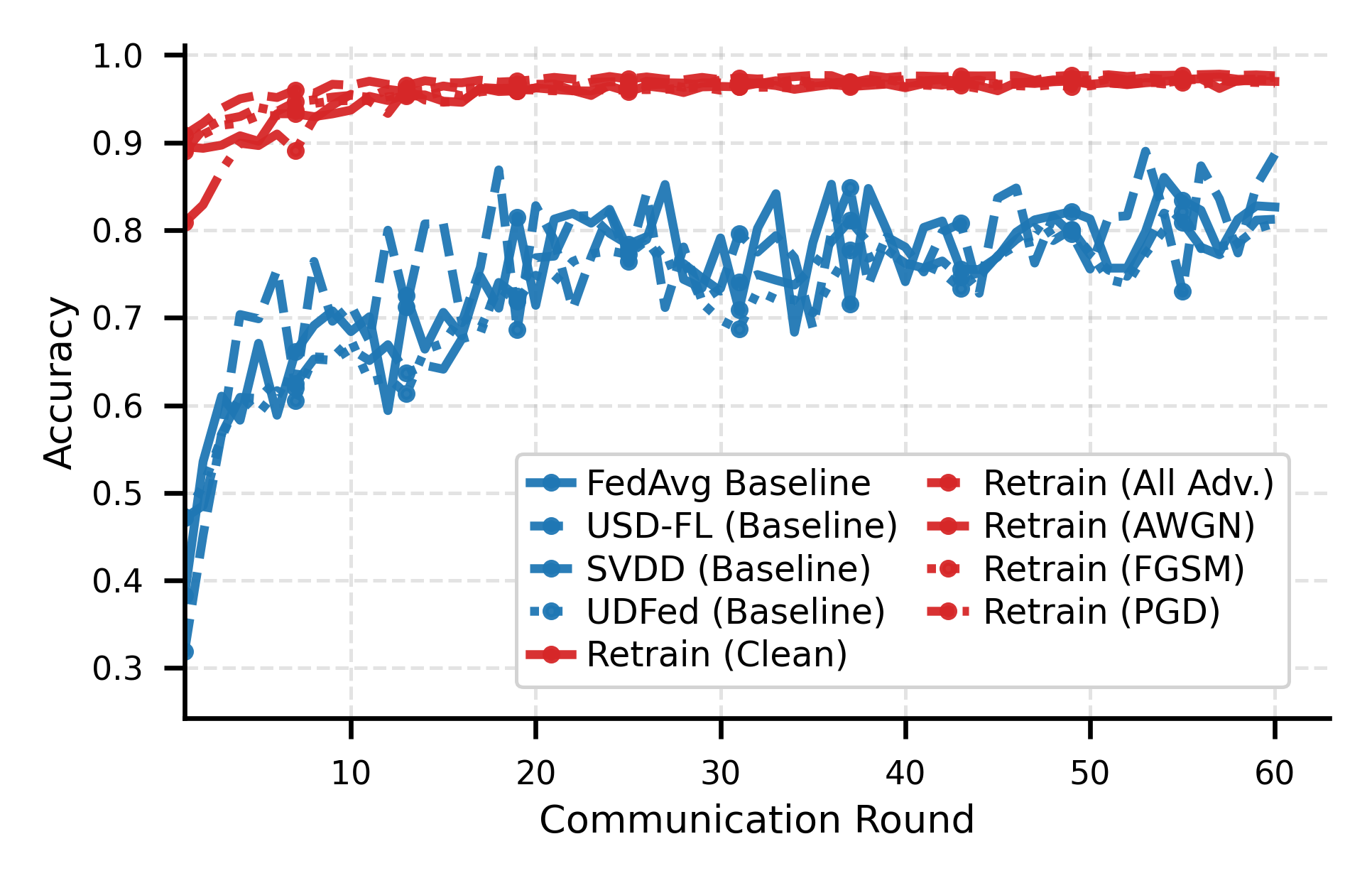}%
    \label{fig:amnist_iid_mixed}}
  \hfill
  \subfloat[\textnormal{\textit{AudioMNIST non-IID}}]{%
    \includegraphics[width=0.49\textwidth]{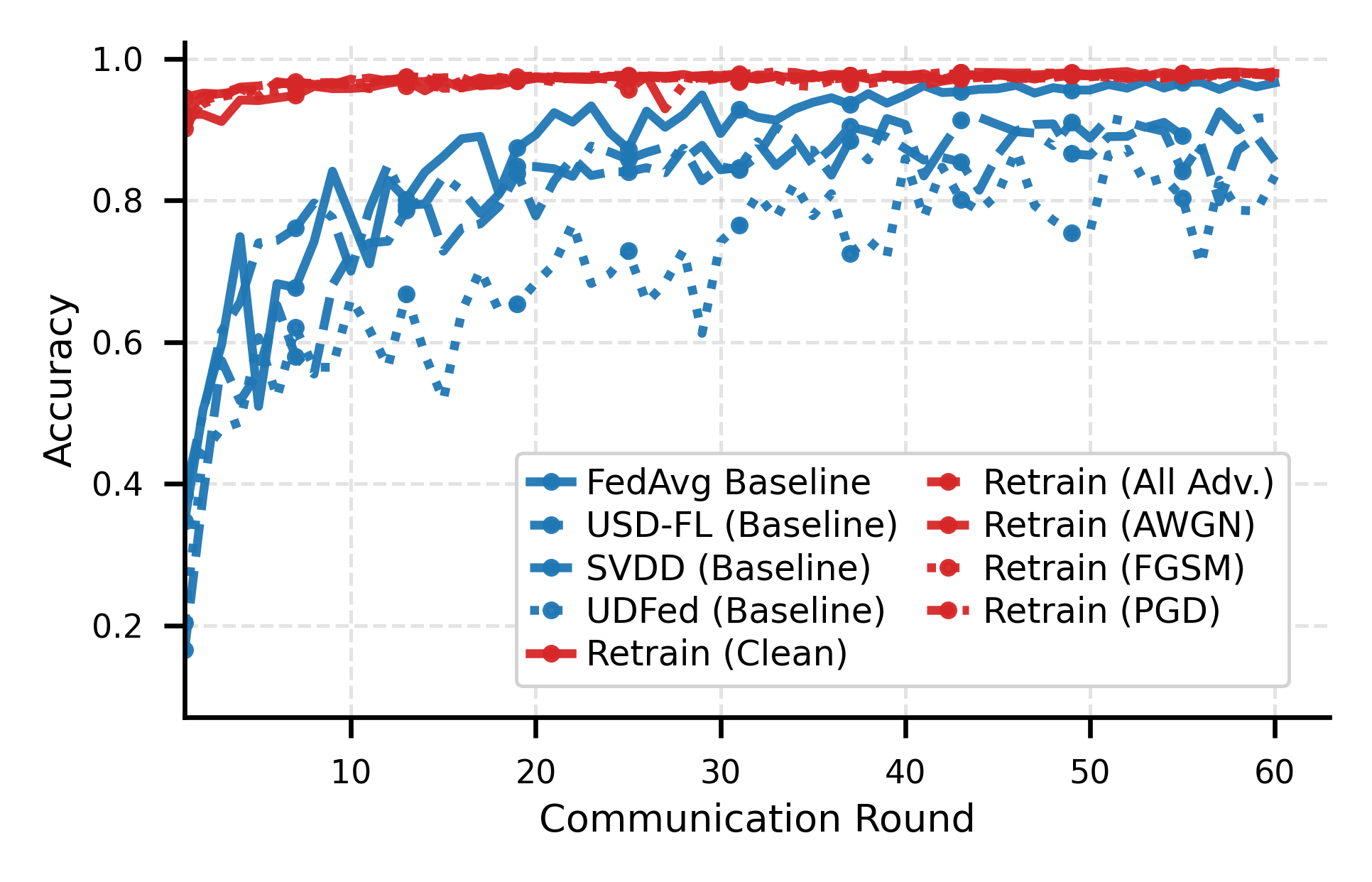}%
    \label{fig:amnist_non_iid_mixed}}

  \vspace{0.6em}

  \subfloat[\textnormal{\textit{UrbanSound8K IID}}]{%
    \includegraphics[width=0.49\textwidth]{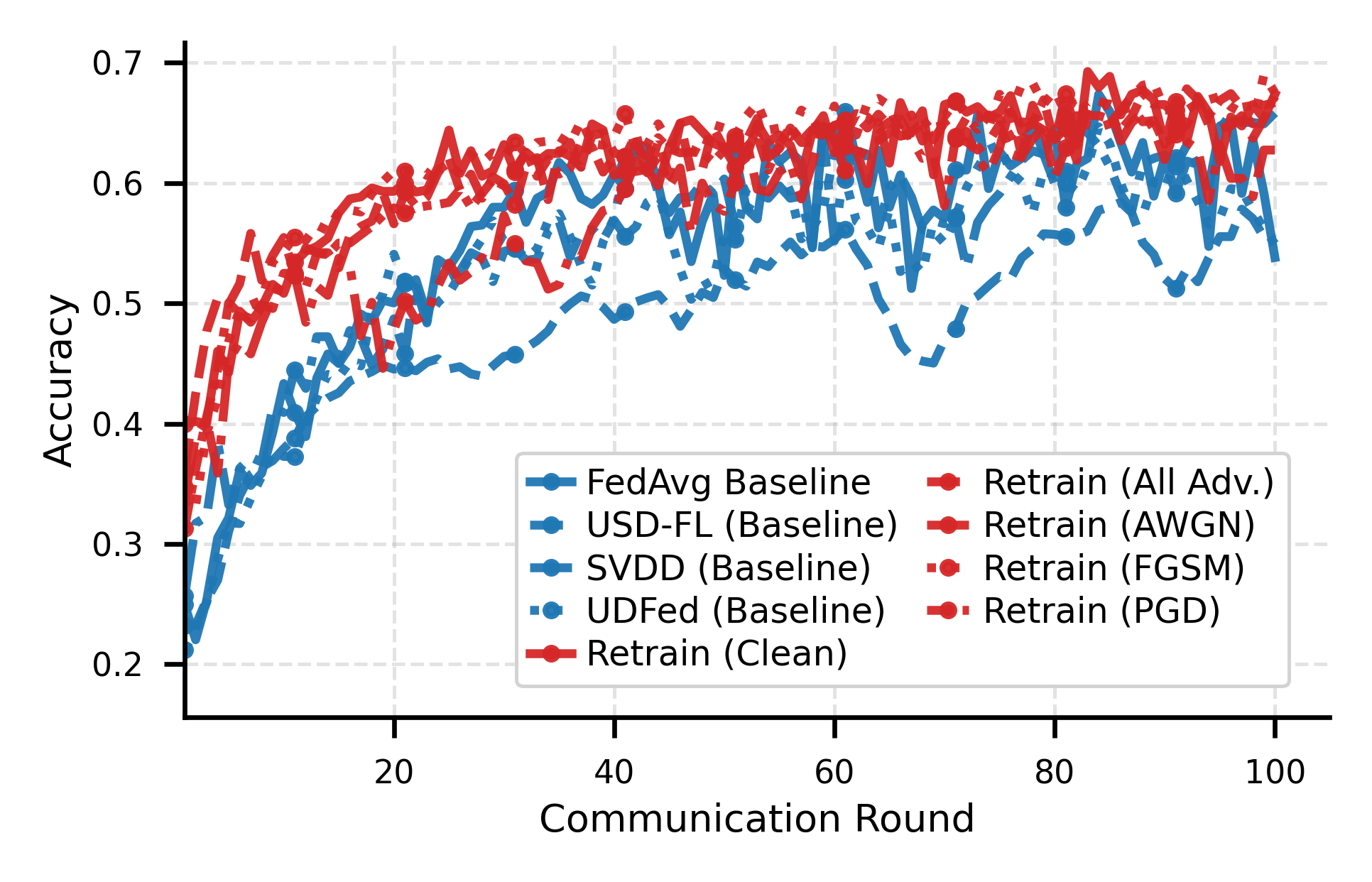}%
    \label{fig:us8k_iid_mixed}}
  \hfill
  \subfloat[\textnormal{\textit{UrbanSound8K non-IID}}]{%
    \includegraphics[width=0.49\textwidth]{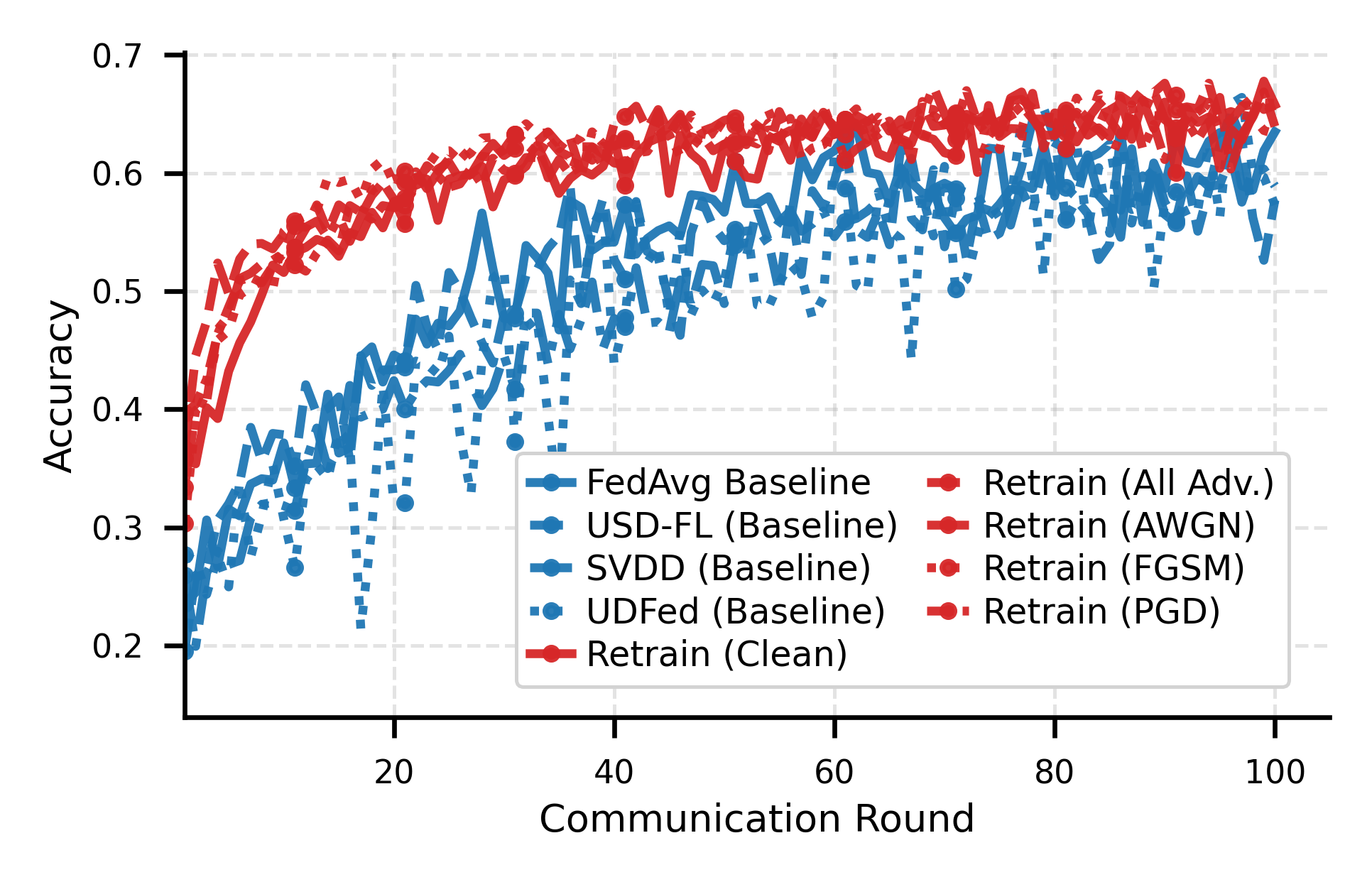}%
    \label{fig:us8k_non_iid_mixed}}
  \caption{Global accuracy over communication rounds under mixed poisoning attacks, where each adversarial client independently and randomly applies FGSM, PGD, or AWGN per round, for (a) AudioMNIST IID, (b) AudioMNIST non-IID, (c) UrbanSound8K IID, and (d) UrbanSound8K non-IID.}
  \label{fig:mixed_results}
\end{figure*}

\begin{table*}[t]
\centering
\caption{10-fold cross-validation results on AudioMNIST IID and non-IID under FGSM poisoning. Mean accuracy over last 5 rounds per fold, reported as mean $\pm$ std with 95\% confidence interval. Statistical significance reported vs. FedAvg (two-sided paired $t$-test).}
\label{tab:kfold_results}
\resizebox{\textwidth}{!}{%
\begin{tabular}{lcccccccccc}
\toprule
& \multicolumn{5}{c}{\textbf{AudioMNIST IID}} & \multicolumn{5}{c}{\textbf{AudioMNIST non-IID}} \\
\cmidrule(lr){2-6} \cmidrule(lr){7-11}
\textbf{Method} & \textbf{Mean} & \textbf{Std} & \textbf{95\% CI} & \textbf{$p$-value} & \textbf{Sig?} & \textbf{Mean} & \textbf{Std} & \textbf{95\% CI} & \textbf{$p$-value} & \textbf{Sig?} \\
\midrule
FedAvg & 82.09\% & 2.43\% & [80.58, 83.60] & -- & -- & 79.88\% & 3.54\% & [77.69, 82.07] & -- & -- \\
USD-FL & 82.72\% & 2.07\% & [81.44, 84.00] & 0.024 & Yes & 83.26\% & 3.41\% & [81.15, 85.37] & 0.018 & Yes \\
UDFed & 83.72\% & 2.28\% & [82.31, 85.13] & 0.193 & No & 79.14\% & 4.02\% & [76.65, 81.63] & 0.393 & No \\
Deep SVDD & 93.08\% & 0.68\% & [92.66, 93.50] & $<$0.0001 & Yes & 89.50\% & 1.48\% & [88.58, 90.42] & $<$0.0001 & Yes \\
\midrule
\multicolumn{11}{l}{\textit{REVERB-FL configurations (ours)}} \\
Retrain (No Poison) & 96.29\% & 0.39\% & [96.05, 96.53] & $<$0.0001 & Yes & 96.30\% & 0.38\% & [96.06, 96.54] & $<$0.0001 & Yes \\
Retrain (FGSM) & 96.54\% & 0.52\% & [96.22, 96.86] & $<$0.0001 & Yes & 96.42\% & 0.52\% & [96.10, 96.74] & $<$0.0001 & Yes \\
Retrain (PGD) & 96.35\% & 0.52\% & [96.03, 96.67] & $<$0.0001 & Yes & 96.32\% & 0.50\% & [96.01, 96.63] & $<$0.0001 & Yes \\
Retrain (AWGN) & 96.55\% & 0.55\% & [96.21, 96.89] & $<$0.0001 & Yes & 96.61\% & 0.53\% & [96.28, 96.94] & $<$0.0001 & Yes \\
Retrain (All Adv.) & 96.40\% & 0.52\% & [96.08, 96.72] & $<$0.0001 & Yes & 96.33\% & 0.50\% & [96.02, 96.64] & $<$0.0001 & Yes \\
\bottomrule
\end{tabular}}
\end{table*}

\subsection{Non-i.i.d. results}\label{non-res}
Client heterogeneity amplifies poisoning effects (Figs.~\ref{fig:amnist_non_iid_2x2}, \ref{fig:us8k_non_iid_2x2}). We evaluate the same baselines under non-IID label skew (Dirichlet $\alpha=0.5$), where each client observes highly imbalanced class distributions.

On \textbf{AudioMNIST}, heterogeneity amplifies baseline instability (Figs.~\ref{fig:amnist_non_iid_2x2}). UDFed degrades severely under all poisoning conditions, while FedAvg and USD-FL show high variance. In contrast, all REVERB-FL configurations converge stably to 96--98\% across all attack conditions, with a variant of our method achieving the best performance in every column (Table~\ref{tab:summary_results}). 

On \textbf{UrbanSound8K}, non-IID conditions widen the gap between reserve methods and baselines relative to IID (Figs.~\ref{fig:us8k_non_iid_2x2}). Baselines degrade to 51--65\% under gradient-based attacks, while REVERB-FL configurations sustain 59--68\%, representing a 10--15\% improvement under PGD poisoning. Across all UrbanSound8K non-IID conditions, a variant of our method consistently achieves the highest accuracy (Table~\ref{tab:summary_results}).

Non-IID conditions amplify both convergence instability and poisoning vulnerability, with performance gaps being most pronounced under PGD poisoning. Attack-specific retraining maximizes defense against known attacks, while \textbf{Retrain (All Adversarial)} provides robust cross-attack performance without requiring knowledge of the attack type, making it suitable for practical deployment where attack types may vary or be unknown.

In addition, these results provide partial evidence against adaptive adversaries: since retraining configurations retain strong performance even under mismatched attack types (e.g., Retrain (FGSM) under PGD poisoning, Figs.~\ref{fig:amnist_iid_2x2},~\ref{fig:amnist_non_iid_2x2}), the defense is not easily circumvented by varying the attack strategy. Future work may further investigate a full evaluation of adversaries explicitly targeting the reserve-set mechanism with additional adaptive client poisoning strategies.

\begin{table*}[t]
\centering
\caption{Summary of final-round accuracy (\%) across all experimental conditions, 
reported as mean accuracy over the last 5 communication rounds. 
\textbf{Bold} indicates the best result per column.}
\label{tab:summary_results}
\resizebox{\textwidth}{!}{%
\begin{tabular}{lcccccccccc}
\toprule
& \multicolumn{5}{c}{\textbf{AudioMNIST IID}} & \multicolumn{5}{c}{\textbf{AudioMNIST non-IID}} \\
\cmidrule(lr){2-6} \cmidrule(lr){7-11}
\textbf{Method} & Clean & FGSM & PGD & AWGN & Mixed & Clean & FGSM & PGD & AWGN & Mixed \\
\midrule
FedAvg                & 94.4 & 84.9 & 85.1 & 94.5 & 12.7 & 97.1 & 94.2 & 96.1 & 95.9 & 96.3 \\
UDFed                 & 94.1 & 81.0 & 76.4 & 84.4 & 96.9 & 90.4 & 79.8 & 80.3 & 81.6 & 78.9 \\
USD-FL                & 94.2 & 94.6 & 94.4 & 94.4 & 84.4 & 91.4 & 87.2 & 93.0 & 93.2 & 85.9 \\
Deep SVDD             & 94.4 & 92.5 & 93.5 & 93.9 & 93.2 & 92.6 & 91.3 & 91.4 & 91.8 & 90.6 \\
\midrule
\multicolumn{11}{l}{\textbf{REVERB-FL configurations}} \\
Retrain (No Poison)   & \textbf{97.4} & 96.8 & 96.8 & 97.1 & 96.9 & \textbf{98.3} & 96.9 & 97.8 & 97.6 & \textbf{97.9} \\
Retrain (FGSM)        & \textbf{97.4} & 96.8 & 96.6 & 96.7 & 96.9 & 97.9 & 96.6 & 97.9 & \textbf{98.0} & 97.8 \\
Retrain (PGD)         & 97.2 & 96.8 & 96.8 & \textbf{97.8} & 97.1 & 98.1 & \textbf{97.5} & 97.9 & 96.9 & 97.8 \\
Retrain (AWGN)        & 97.0 & 96.5 & 97.1 & 97.2 & 97.4 & 98.2 & 97.2 & 97.9 & \textbf{98.0} & 97.8 \\
Retrain (All Adv.)    & 97.1 & \textbf{96.9} & \textbf{97.3} & 96.8 & \textbf{97.7} & 98.0 & 97.3 & \textbf{98.1} & \textbf{98.0} & 97.7 \\
\midrule
& \multicolumn{5}{c}{\textbf{UrbanSound8K IID}} & \multicolumn{5}{c}{\textbf{UrbanSound8K non-IID}} \\
\cmidrule(lr){2-6} \cmidrule(lr){7-11}
\textbf{Method} & Clean & FGSM & PGD & AWGN & Mixed & Clean & FGSM & PGD & AWGN & Mixed \\
\midrule
FedAvg                & 70.8 & 55.2 & 53.6 & 63.4 & 58.4 & 67.2 & 52.4 & 58.3 & 60.1 & 61.3 \\
UDFed                 & 69.3 & 57.8 & 58.2 & 56.4 & 69.4 & 66.8 & 57.2 & 52.1 & 57.3 & 61.9 \\
USD-FL                & 70.1 & 58.6 & 52.4 & 64.2 & 22.0 & 68.9 & 55.1 & 54.2 & 58.3 & 58.9 \\
Deep SVDD             & 69.4 & 60.1 & 62.3 & 63.8 & 65.3 & 67.4 & 58.3 & 51.6 & 60.4 & 59.7 \\
\midrule
\multicolumn{11}{l}{\textbf{REVERB-FL configurations}} \\
Retrain (No Poison)   & 71.7 & 67.5 & 60.6 & 65.3 & 65.9 & \textbf{69.4} & 63.5 & 59.0 & 63.6 & 65.0 \\
Retrain (FGSM)        & 68.2 & 67.0 & 65.8 & 67.0 & 66.7 & \textbf{69.4} & \textbf{66.0} & 63.5 & 62.8 & \textbf{65.6} \\
Retrain (PGD)         & \textbf{71.8} & \textbf{68.6} & \textbf{68.0} & \textbf{67.6} & 66.7 & 69.0 & 65.7 & \textbf{63.9} & \textbf{68.4} & 65.1 \\
Retrain (AWGN)        & 70.1 & 61.2 & 63.3 & 62.8 & 61.0 & 68.0 & 64.7 & 62.2 & 63.9 & 63.0 \\
Retrain (All Adv.)    & 69.6 & 64.4 & 66.8 & 67.1 & \textbf{66.7} & 67.7 & 64.1 & 62.4 & 63.7 & 65.5 \\
\bottomrule
\end{tabular}}
\end{table*}

To further assess robustness under heterogeneous adversarial behavior, we evaluate a setting in which each adversarial client independently and randomly applies one of FGSM, PGD, or AWGN per round with no fixed strategy. Results across both datasets and partitions (Fig.~\ref{fig:mixed_results}) show that all REVERB-FL configurations maintain strong robustness, with baselines degrading significantly. On AudioMNIST non-IID, REVERB-FL methods converge stably to 95--97\% while baselines exhibit instability at 80--93\%. On UrbanSound8K, REVERB-FL sustains 63--68\% compared to 54--63\% for baselines. Notably, all REVERB-FL configurations perform comparably, demonstrating that the reserve retraining mechanism provides robustness regardless of the specific augmentation strategy used at the server.

Finally, to assess statistical stability, we conduct 10-fold cross-validation on AudioMNIST under FGSM poisoning for both IID and non-IID partitions, including all baselines and all REVERB-FL configurations. Results are summarized in Table~\ref{tab:kfold_results}. All REVERB-FL configurations achieve mean accuracy of 96.3--96.6\% with low standard deviation (0.39--0.55\%) under both IID and non-IID conditions, compared to FedAvg at 82.1\% $\pm$ 2.43\% (IID) and 79.9\% $\pm$ 3.54\% (non-IID). All REVERB-FL improvements are statistically significant ($p < 0.0001$, two-sided paired $t$-test vs.\ FedAvg), with no overlap between FedAvg and any REVERB-FL configuration across all 10 folds. Among baselines, Deep SVDD achieves the strongest performance at 93.1\% (IID) and 89.5\% (non-IID), but remains 3.2--3.5\% and 6.8--7.1\% below the weakest IID and non-IID REVERB-FL configuration, respectively. Critically, the 95\% confidence intervals of all REVERB-FL configurations do not overlap with those of any baseline in either setting, providing strong evidence that the improvements are statistically robust across data splits. Notably, UDFed (baseline) does not achieve statistical significance vs. FedAvg (baseline) in either setting ($p = 0.19$ IID, $p = 0.39$ non-IID), suggesting its improvements are not consistent across data splits. Full cross-validation results for other datasets were consistent with the trends shown in Table \ref{tab:kfold_results} and, thus, were omitted for brevity.

\section{Limitations} \label{sec:limitations}
Several limitations of the current framework warrant discussion. First, REVERB-FL operates under a closed-set assumption in which the global label space is fixed and shared across all clients and the server reserve set. Thus, our framework, along with standard federated audio classification settings, needs to be adapted to open-set label spaces prior to being applied in open-set scenarios where clients may encounter or inject unknown-class audio samples. For such scenarios, the literature warrants a dynamic framework capable of adapting to open-set and out of distribution label spaces. Second, an adversarial client could potentially inject out-of-distribution samples with in-distribution labels to poison the training process. Due to the effectiveness of our method in mitigating general poisoning attacks, reserve-set retraining on clean, representative data (as done at the conclusion of each training iteration in REVERB-FL) would partially counteract such corruption by mitigating drifts introduced by the poisoned triggers. However, it is important to note that the extent to which such types of attacks would be mitigated by REVERB-FL is beyond the scope of our work. Third, the reserve set is assumed to be clean and trusted. In practice, collecting a perfectly clean reserve set may be challenging, and performance may degrade if $\mathcal{D}_r$ is contaminated or severely imbalanced. Since $\mathcal{D}_r$ is small and collected once before training, lightweight quality validation prior to deployment can mitigate this risk; furthermore, the bounded impact of mild imbalance is captured theoretically by the mismatch term $\varepsilon_r$ in Proposition~2, which remains small under stratified sampling. Fourth, the fixed adversarial fraction assumption may not reflect real-world settings where adversarial participation varies across rounds.

\section{Conclusion} \label{conclusion_sec}
This work introduced \textbf{REVERB-FL}, a server-side defense framework for federated audio classification that integrates reserve-set retraining with adversarial augmentation. Our approach enhances robustness to poisoning and non-IID heterogeneity without modifying the client-side protocol or aggregation rule. Through experiments in a multitude of settings on \emph{AudioMNIST} and \emph{UrbanSound8K} under various model poisoning perturbations, REVERB-FL consistently improved convergence stability and maintained higher global accuracy compared to baseline FedAvg and existing poison defenses. Our theoretical analysis established a round-wise contraction bound, demonstrating accelerated convergence and reduced steady-state error in the presence of adversarial poisoning attacks.

By leveraging a small trusted subset at the server, REVERB-FL achieves robustness while preserving data privacy and scalability, making it directly compatible with standard FL frameworks. Future work will focus on (i) integrating dynamic reserve selection to adapt to varying attack intensities, (ii) exploring adaptive aggregation or certified robustness bounds, or (iii) testing with variable clients and client data. Moreover, future work will consider applications of REVERB-FL in domains beyond audio for secure signal-domain federated learning, such as RF sensing, biomedical signals, image processing, and industrial acoustics. Overall, REVERB-FL introduced a  privacy-preserving federated learning and adversarially resilient audio signal modeling framework, offering a theoretically-backed and practical direction for robust audio FL systems.
\bibliography{references}

@ARTICLE{fl_survey,
  author={Kumar, Kummari Naveen and Mohan, Chalavadi Krishna and Cenkeramaddi, Linga Reddy},
  journal={IEEE Transactions on Pattern Analysis and Machine Intelligence}, 
  title={The Impact of Adversarial Attacks on Federated Learning: A Survey}, 
  year={2024},
  volume={46},
  number={5},
  pages={2672-2691},
  doi={10.1109/TPAMI.2023.3322785}}

@misc{fl_speech_emotion,
      title={Robust Federated Learning Against Adversarial Attacks for Speech Emotion Recognition}, 
      author={Yi Chang and Sofiane Laridi and Zhao Ren and Gregory Palmer and Björn W. Schuller and Marco Fisichella},
      year={2022},
      eprint={2203.04696},
      archivePrefix={arXiv},
      primaryClass={cs.SD},
      url={https://arxiv.org/abs/2203.04696}, 
}

@ARTICLE{robust_audio_adversarial,
  author={Esmaeilpour, Mohammad and Cardinal, Patrick and Lameiras Koerich, Alessandro},
  journal={IEEE Transactions on Information Forensics and Security}, 
  title={A Robust Approach for Securing Audio Classification Against Adversarial Attacks}, 
  year={2020},
  volume={15},
  number={},
  pages={2147-2159},
  doi={10.1109/TIFS.2019.2956591}}

@ARTICLE{evasion_attacks_fl,
  author={Wang, Su and Sahay, Rajeev and Piaseczny, Adam and Brinton, Christopher G.},
  journal={IEEE Transactions on Network Science and Engineering}, 
  title={Mitigating Evasion Attacks in Federated Learning Based Signal Classifiers}, 
  year={2025},
  volume={12},
  number={5},
  pages={3933-3947},
  doi={10.1109/TNSE.2025.3566954}}

@misc{fgsm,
      title={Explaining and Harnessing Adversarial Examples}, 
      author={Ian J. Goodfellow and Jonathon Shlens and Christian Szegedy},
      year={2015},
      eprint={1412.6572},
      archivePrefix={arXiv},
      primaryClass={stat.ML},
      url={https://arxiv.org/abs/1412.6572}, 
}

@misc{pgd,
      title={Towards Deep Learning Models Resistant to Adversarial Attacks}, 
      author={Aleksander Madry and Aleksandar Makelov and Ludwig Schmidt and Dimitris Tsipras and Adrian Vladu},
      year={2019},
      eprint={1706.06083},
      archivePrefix={arXiv},
      primaryClass={stat.ML},
      url={https://arxiv.org/abs/1706.06083}, 
}

@ARTICLE{deep_svdd_cite,
  author={Zhang, Anqi and Zhao, Ping and Lu, Wenke and Zhou, Yipeng and Zhang, Wenqian and Zhang, Guanglin},
  journal={IEEE Transactions on Cognitive Communications and Networking}, 
  title={Mitigating Poisoning Attacks in Federated Learning Through Deep One-Class Classification}, 
  year={2025},
  volume={},
  number={},
  pages={1-1},
  doi={10.1109/TCCN.2025.3564476}}

@ARTICLE{udfed_cite,
  author={Deng, Jieyi and Li, Congduan and Zhang, Nanfeng and Yang, Jingfeng and Gao, Jun},
  journal={IEEE Transactions on Information Forensics and Security}, 
  title={UDFed: A Universal Defense Scheme for Various Poisoning Attacks on Federated Learning}, 
  year={2025},
  volume={20},
  number={},
  pages={10480-10494},
  doi={10.1109/TIFS.2025.3611126}}

@article{wang2020tacklingobjectiveinconsistencyproblem,
  title={Tackling the objective inconsistency problem in heterogeneous federated optimization},
  author={Wang, Jianyu and Liu, Qinghua and Liang, Hao and Joshi, Gauri and Poor, H Vincent},
  journal={Advances in neural information processing systems},
  volume={33},
  pages={7611--7623},
  year={2020}
}

@article{Rodriguez-Barroso2022Survey,
title={Survey on Federated Learning Threats: concepts, taxonomy on attacks and defences, experimental study and challenges},
author={Nuria Rodr'iguez-Barroso and Daniel Jim'enez L'opez and M. V. Luz'on and Francisco Herrera and Eugenio Martínez-Cámara},
journal={ArXiv},
year={2022},
volume={abs/2201.08135},
doi={10.1016/j.inffus.2022.09.011}
}

@article{becker2023audiomnistexploringexplainableartificial,
  title={Audiomnist: Exploring explainable artificial intelligence for audio analysis on a simple benchmark},
  author={Becker, S{\"o}ren and Vielhaben, Johanna and Ackermann, Marcel and M{\"u}ller, Klaus-Robert and Lapuschkin, Sebastian and Samek, Wojciech},
  journal={Journal of the Franklin Institute},
  volume={361},
  number={1},
  pages={418--428},
  year={2024},
  publisher={Elsevier}
}

@inproceedings{urbansound,
author = {Salamon, Justin and Jacoby, Christopher and Bello, Juan Pablo},
title = {A Dataset and Taxonomy for Urban Sound Research},
year = {2014},
booktitle = {Proceedings of the 22nd ACM International Conference on Multimedia},
pages = {1041–1044}
}

@INPROCEEDINGS{dirichlet,
  author={Bouguila, Nizar and Ziou, Djemel},
  booktitle={2008 IEEE Workshop on Machine Learning for Signal Processing}, 
  title={A Dirichlet process mixture of dirichlet distributions for classification and prediction}, 
  year={2008},
  volume={},
  number={},
  pages={297-302},
  doi={10.1109/MLSP.2008.4685496}}

@ARTICLE{stft,
  author={Allen, J.},
  journal={IEEE Transactions on Acoustics, Speech, and Signal Processing}, 
  title={Short term spectral analysis, synthesis, and modification by discrete Fourier transform}, 
  year={1977},
  volume={25},
  number={3},
  pages={235-238},
  doi={10.1109/TASSP.1977.1162950}}

@INPROCEEDINGS{spectrograms,
  author={Piczak, Karol J.},
  booktitle={2015 IEEE 25th International Workshop on Machine Learning for Signal Processing (MLSP)}, 
  title={Environmental sound classification with convolutional neural networks}, 
  year={2015},
  volume={},
  number={},
  pages={1-6},
  doi={10.1109/MLSP.2015.7324337}}

@inproceedings{yan2025federated,
  title={Federated Adversarial Defense with Adversarial Training and Personalized Evaluation},
  author={Yan, Liwei and Zhu, Qi and Zhai, Xiangping},
  booktitle={2025 2nd International Conference on Digital Media, Communication and Information Systems (DMCIS)},
  pages={121--124},
  year={2025}
}

@ARTICLE{fl_poisoning_attack,
author={Campos, Enrique Marmol and Gonzalez-Vidal, Aurora and Hernandez-Ramos, Jose L. and Skarmeta, Antonio},
journal={ IEEE Transactions on Emerging Topics in Computing },
title={{ FedRDF: A Robust and Dynamic Aggregation Function Against Poisoning Attacks in Federated Learning }},
year={2025},
volume={13},
number={01},
ISSN={2168-6750},
pages={48-67}
}

@INPROCEEDINGS{zhang2022transformers,
  author={Zhang, Yixiao and Li, Baihua and Fang, Hui and Meng, Qinggang},
  booktitle={2022 IEEE International Conference on Imaging Systems and Techniques (IST)}, 
  title={Spectrogram Transformers for Audio Classification}, 
  year={2022},
  volume={},
  number={},
  pages={1-6},
  doi={10.1109/IST55454.2022.9827729}}

@ARTICLE{wang2022selfsupervised,
  author={Wang, Shanshan and Politis, Archontis and Mesaros, Annamaria and Virtanen, Tuomas},
  journal={IEEE Journal of Selected Topics in Signal Processing}, 
  title={Self-Supervised Learning of Audio Representations From Audio-Visual Data Using Spatial Alignment}, 
  year={2022},
  volume={16},
  number={6},
  pages={1467-1479},
  doi={10.1109/JSTSP.2022.3180592}}

@INPROCEEDINGS{fl_speech_recognition,
  author={Tan, Chao and Cao, Yang and Li, Sheng and Yoshikawa, Masatoshi},
  booktitle={2023 IEEE International Conference on Acoustics, Speech and Signal Processing (ICASSP)}, 
  title={General or Specific? Investigating Effective Privacy Protection in Federated Learning for Speech Emotion Recognition}, 
  year={2023},
  volume={},
  number={},
  pages={1-5},
  doi={10.1109/ICASSP49357.2023.10096844}}

@INPROCEEDINGS{fl_kws,
  author={Grollmisch, Sascha and Köllmer, Thomas and Yaroshchuk, Artem and Lukashevich, Hanna},
  booktitle={2025 IEEE International Conference on Acoustics, Speech, and Signal Processing Workshops (ICASSPW)}, 
  title={Federated Semi-supervised Learning for Industrial Sound Analysis and Keyword Spotting}, 
  year={2025},
  volume={},
  number={},
  pages={1-5},
  doi={10.1109/ICASSPW65056.2025.11011203}}

@INPROCEEDINGS{dai2017very,
  author={Dai, Wei and Dai, Chia and Qu, Shuhui and Li, Juncheng and Das, Samarjit},
  booktitle={2017 IEEE International Conference on Acoustics, Speech and Signal Processing (ICASSP)}, 
  title={Very deep convolutional neural networks for raw waveforms}, 
  year={2017},
  volume={},
  number={},
  pages={421-425},
  doi={10.1109/ICASSP.2017.7952190}}

@misc{hsu2019measuring,
      title={Measuring the Effects of Non-Identical Data Distribution for Federated Visual Classification}, 
      author={Tzu-Ming Harry Hsu and Hang Qi and Matthew Brown},
      year={2019},
      eprint={1909.06335},
      archivePrefix={arXiv},
      primaryClass={cs.LG},
      url={https://arxiv.org/abs/1909.06335}, 
}

@ARTICLE{fl_audio_challenges,
  author={Gafni, Tomer and Shlezinger, Nir and Cohen, Kobi and Eldar, Yonina C. and Poor, H. Vincent},
  journal={IEEE Signal Processing Magazine}, 
  title={Federated Learning: A signal processing perspective}, 
  year={2022},
  volume={39},
  number={3},
  pages={14-41},
  doi={10.1109/MSP.2021.3125282}}

@article{fedprox,
  title={Federated optimization in heterogeneous networks},
  author={Li, Tian and Sahu, Anit Kumar and Zaheer, Manzil and Sanjabi, Maziar and Talwalkar, Ameet and Smith, Virginia},
  journal={Proceedings of Machine learning and systems},
  volume={2},
  pages={429--450},
  year={2020}
}

@inproceedings{per_fedavg,
author = {Fallah, Alireza and Mokhtari, Aryan and Ozdaglar, Asuman},
title = {Personalized federated learning with theoretical guarantees: a model-agnostic meta-learning approach},
year = {2020},
booktitle = {Proceedings of the 34th International Conference on Neural Information Processing Systems},
articleno = {300}
}

@InProceedings{scaffold,
  title = 	 {{SCAFFOLD}: Stochastic Controlled Averaging for Federated Learning},
  author =       {Karimireddy, Sai Praneeth and Kale, Satyen and Mohri, Mehryar and Reddi, Sashank and Stich, Sebastian and Suresh, Ananda Theertha},
  booktitle = 	 {Proceedings of the 37th International Conference on Machine Learning},
  pages = 	 {5132--5143},
  year = 	 {2020},
  volume = 	 {119}
}

@inproceedings{pfedme,
author = {Dinh, Canh T. and Tran, Nguyen H. and Nguyen, Tuan Dung},
title = {Personalized federated learning with moreau envelopes},
year = {2020},
booktitle = {Proceedings of the 34th International Conference on Neural Information Processing Systems},
articleno = {1796}
}

@inproceedings{krum,
author = {Blanchard, Peva and El Mhamdi, El Mahdi and Guerraoui, Rachid and Stainer, Julien},
title = {Machine learning with adversaries: byzantine tolerant gradient descent},
year = {2017},
booktitle = {Proceedings of the 31st International Conference on Neural Information Processing Systems},
pages = {118–128}
}

@inproceedings{bulyan,
  title={The hidden vulnerability of distributed learning in byzantium},
  author={El Mahdi El Mhamdi and Rachid Guerraoui and Sébastien Rouault},
  booktitle={International conference on machine learning},
  pages={3521--3530},
  year={2018},
  organization={PMLR}
}

@InProceedings{trimmedmean,
  title = 	 {Byzantine-Robust Distributed Learning: Towards Optimal Statistical Rates},
  author =       {Yin, Dong and Chen, Yudong and Kannan, Ramchandran and Bartlett, Peter},
  booktitle = 	 {Proceedings of the 35th International Conference on Machine Learning},
  pages = 	 {5650--5659},
  year = 	 {2018}
}

@inproceedings{median,
  title={Draco: Byzantine-resilient distributed training via redundant gradients},
  author={Chen, Lingjiao and Wang, Hongyi and Charles, Zachary and Papailiopoulos, Dimitris},
  booktitle={International Conference on Machine Learning},
  pages={903--912},
  year={2018}
}

@misc{tramer2018ensemble,
      title={Ensemble Adversarial Training: Attacks and Defenses}, 
      author={Florian Tramèr and Alexey Kurakin and Nicolas Papernot and Ian Goodfellow and Dan Boneh and Patrick McDaniel},
      year={2020},
      eprint={1705.07204},
      archivePrefix={arXiv},
      primaryClass={stat.ML},
      url={https://arxiv.org/abs/1705.07204}, 
}

@INPROCEEDINGS{server_finetune,
  author={He, Xuechao and Zhu, Heng and Ling, Qing},
  booktitle={2022 IEEE International Conference on Acoustics, Speech and Signal Processing (ICASSP)}, 
  title={Byzantine-Robust and Communication-Efficient Distributed Non-Convex Learning Over Non-IID Data}, 
  year={2022},
  pages={5223-5227}
}

@INPROCEEDINGS{defense_aware_fl,
  author={Yi, Liping and Shi, Xiaorong and Wang, Wenrui and Wang, Gang and Liu, Xiaoguang},
  booktitle={2023 International Joint Conference on Neural Networks (IJCNN)}, 
  title={FedRRA: Reputation-Aware Robust Federated Learning against Poisoning Attacks}, 
  year={2023},
  volume={},
  number={},
  pages={1-8},
  doi={10.1109/IJCNN54540.2023.10191556}}

@article{fl_at_instability,
title = {Adversarial Machine Learning in Industry: A Systematic Literature Review},
journal = {Computers \& Security},
volume = {145},
pages = {103988},
year = {2024},
author = {Felix Viktor Jedrzejewski and Lukas Thode and Jannik Fischbach and Tony Gorschek and Daniel Mendez and Niklas Lavesson}
}

@article{2018federated,
  author       = {Andrew Hard and
                  Kanishka Rao and
                  Rajiv Mathews and
                  Fran{\c{c}}oise Beaufays and
                  Sean Augenstein and
                  Hubert Eichner and
                  Chlo{\'{e}} Kiddon and
                  Daniel Ramage},
  title        = {Federated Learning for Mobile Keyboard Prediction},
  journal      = {CoRR},
  volume       = {abs/1811.03604},
  year         = {2018},
  url          = {http://arxiv.org/abs/1811.03604},
  eprinttype    = {arXiv},
  eprint       = {1811.03604},
  bibsource    = {dblp computer science bibliography, https://dblp.org}
}

@INPROCEEDINGS{knowledge_distillation,
  author={Chen, Yu-Wen and Ke, Bo-Hsu and Chen, Bo-Zhong and Chiu, Si-Rong and Tu, Chun-Wei and Kuo, Jian-Jhih},
  booktitle={2023 IEEE Global Communications Conference}, 
  title={Knowledge Distillation Based Defense for Audio Trigger Backdoor in Federated Learning}, 
  year={2023},
  volume={},
  number={},
  pages={4271-4276},
  doi={10.1109/GLOBECOM54140.2023.10437601}}

@article{dinh2022new,
  title={A new look and convergence rate of federated multitask learning with laplacian regularization},
  author={Dinh, Canh T and Vu, Tung T and Tran, Nguyen H and Dao, Minh N and Zhang, Hongyu},
  journal={IEEE Transactions on Neural Networks and Learning Systems},
  volume={35},
  number={6},
  pages={8075--8085},
  year={2022},
  publisher={IEEE}
}

@inproceedings{li2020convergencefedavgnoniiddata,
  author       = {Xiang Li and
                  Kaixuan Huang and
                  Wenhao Yang and
                  Shusen Wang and
                  Zhihua Zhang},
  title        = {On the Convergence of FedAvg on Non-IID Data},
  booktitle    = {8th International Conference on Learning Representations, 2020}
}

@article{zhao2018federated,
  title={Federated learning with non-iid data},
  author={Zhao, Yue and Li, Meng and Lai, Liangzhen and Suda, Naveen and Civin, Damon and Chandra, Vikas},
  journal={arXiv preprint arXiv:1806.00582},
  year={2018}
}

@inproceedings{fedAvg_baseline,
  title={Communication-efficient learning of deep networks from decentralized data},
  author={McMahan, Brendan and Moore, Eider and Ramage, Daniel and Hampson, Seth and y Arcas, Blaise Aguera},
  booktitle={Artificial intelligence and statistics},
  pages={1273--1282},
  year={2017}
}

@InProceedings{bhagoji2019analyzing,
  title = 	 {Analyzing Federated Learning through an Adversarial Lens},
  author =       {Bhagoji, Arjun Nitin and Chakraborty, Supriyo and Mittal, Prateek and Calo, Seraphin},
  booktitle = 	 {Proceedings of the 36th International Conference on Machine Learning},
  pages = 	 {634--643},
  year = 	 {2019}
}

@inproceedings{poison_fl,
  title={How potent are evasion attacks for poisoning federated learning-based signal classifiers?},
  author={Wang, Su and Sahay, Rajeev and Brinton, Christopher G},
  booktitle={ICC 2023-IEEE International Conference on Communications},
  pages={2376--2381},
  year={2023}
}

@ARTICLE{auditory_perception_disease,
  author={Wang, Zhihua and Zhang, Haojie and Tan, Yang and Wang, Rui and Qian, Kun and Hu, Bin and Yamamoto, Yoshiharu and Schuller, Björn W.},
  journal={IEEE Journal of Biomedical and Health Informatics}, 
  title={Can Information Representations Inspired by the Human Auditory Perception Benefit Computer Audition-Based Disease Detection? An Interpretable Comparative Study}, 
  year={2025},
  volume={},
  number={},
  pages={1-14},
  doi={10.1109/JBHI.2025.3638846}}

@article{heart_sound_interpretable,
title = {Exploring interpretable representations for heart sound abnormality detection},
journal = {Biomedical Signal Processing and Control},
volume = {82},
pages = {104569},
year = {2023},
issn = {1746-8094},
doi = {https://doi.org/10.1016/j.bspc.2023.104569},
url = {https://www.sciencedirect.com/science/article/pii/S1746809423000022},
author = {Zhihua Wang and Kun Qian and Houguang Liu and Bin Hu and Björn W. Schuller and Yoshiharu Yamamoto}
}

\bibliographystyle{IEEEtran}

\appendices
\section{Proof of Theorem 1}
\label{appendix:proofs}

\noindent\textbf{Notation recap.}
The global objective is $\varphi(\theta)=\tfrac{1}{N}\sum_{n=1}^N \varphi_n(\theta)$ with optimal minimizer value $\varphi^\star=\min_\theta \varphi(\theta)$. In round $t$, we denote $\theta_t^{+} \equiv \theta^{(t+1,0)}$ as the 
post-FedAvg iterate before reserve retraining. Clients use $\tau$ local steps; the effective FedAvg stepsize is $\gamma_g$ (instantiated as $\eta\tau$ in our implementation). The reserve performs $r$ server-side steps of size $\gamma_r$. Assumptions in the convergence analysis hold: $L$-smoothness, $\mu$-strong convexity, bounded stochastic variance $\sigma_g^2$, bounded drift $\zeta^2$, and the fixed adversarial client set $A$ with fraction $\rho$. At round $t$, $S_t$ are the $m$ sampled clients, and $\beta_t=|S_t\cap A|/m$ satisfies $[\beta_t]=\rho$ and $[\beta_t^2]=\rho^2+\rho(1-\rho)\frac{N-m}{m(N-1)}$. Each adversarial client may shift its local gradient by at most $\Gamma$.

\subsection{Preliminaries}

\begin{lemma}[Descent lemma]\label{lem:descent}
If $\varphi$ is $L$-smooth, then for any $\theta$, direction $g$, and stepsize $\gamma\le 1/L$,
\begin{equation}
\varphi(\theta-\gamma g)
\;\le\;
\varphi(\theta) - \gamma\langle \nabla\varphi(\theta), g\rangle
+ \tfrac{L\gamma^2}{2}\,\|g\|^2.
\label{eq:descent}
\end{equation}
\end{lemma}

\begin{lemma}[PL inequality under strong convexity]\label{lem:pl}
If $\varphi$ is $\mu$-strongly convex, then for all $\theta$,
\begin{equation}
\|\nabla \varphi(\theta)\|^2 \;\ge\; 2\mu\,\big(\varphi(\theta)-\varphi^\star\big).
\label{eq:pl}
\end{equation}
\end{lemma}

\begin{lemma}[Exact gradient step contracts]\label{lem:contract}
Under $L$-smoothness and $\mu$-strong convexity, any exact step of size $\gamma\in(0,1/L]$ obeys
\begin{equation}
\varphi\!\left(\theta-\gamma\nabla\varphi(\theta)\right)-\varphi^\star
\;\le\;
(1-\mu\gamma)\,\big(\varphi(\theta)-\varphi^\star\big).
\label{eq:exact-contract}
\end{equation}
\end{lemma}
\begin{proof}
Apply \eqref{eq:descent} with $g=\nabla\varphi(\theta)$ to get
$\varphi(\theta-\gamma\nabla\varphi(\theta)) \le \varphi(\theta) - \gamma(1-\tfrac{L\gamma}{2})\|\nabla\varphi(\theta)\|^2$.
Use \eqref{eq:pl} and $\gamma\le 1/L$ (so $1-L\gamma/2\ge 1/2$) to obtain \eqref{eq:exact-contract}.
\end{proof}

\subsection{FedAvg as an inexact gradient step}

Let the aggregated FedAvg direction be
\begin{equation}
\tilde{g}_t \;=\; \nabla\varphi(\theta_t) + e_t,
\qquad
e_t \;=\; \underbrace{\xi_t}_{\text{stochastic}} + \underbrace{d_t}_{\text{drift}} + \underbrace{b_t}_{\text{poisoning}}.
\label{eq:dir-decomp}
\end{equation}

\begin{lemma}[Second moment of the aggregation error]\label{lem:delta}
Under Assumptions~\ref{ass:variance}--\ref{ass:drift}, there exist constants $c_s, c_\tau > 0$ such that
\begin{equation}
\mathbb{E}[\|e_t\|^2]
\;\le\;
\tfrac{c_s}{m}\,\sigma_g^2
\;+\;
c_\tau\,\zeta^2
\;+\;
\mathbb{E}[\beta_t^2]\,\Gamma^2.
\label{eq:delta-second-moment}
\end{equation}
\end{lemma}
\begin{proof}
Using the algebraic decomposition (as stated in the body),\\[5pt]
\begin{math}
\begin{aligned}
\nabla \varphi_n(\theta_n^{(t,j)}) =
\nabla \varphi(\theta_t) &+
\underbrace{\big(\nabla \varphi_n(\theta_n^{(t,j)}) - \nabla \varphi_n(\theta_t)\big)}_{\text{stochastic noise}}\\ &+ \underbrace{\big(\nabla \varphi_n(\theta_t) - \nabla \varphi(\theta_t)\big)}_{\text{client drift}}.
\end{aligned}
\end{math}\\[5pt]
Averaging $m$ clients reduces stochastic variance by $1/m$ (by standard variance averaging), giving $c_s\sigma_g^2/m$ where $c_s$ is an absolute constant. The accumulation of $\tau$ local steps inflates drift, with $c_\tau$ growing as $\tau^2$ (standard in local-SGD analyses \cite{li2020convergencefedavgnoniiddata}). Since only a $\beta_t$ fraction of the $m$ selected clients are adversarial, the aggregate poisoning bias has second moment bounded by $\mathbb{E}[\beta_t^2]\,\Gamma^2$. Summing the three contributions gives \eqref{eq:delta-second-moment}.
\end{proof}

\begin{proposition}
    [FedAvg half-step]\label{prop:half}
    Let $\gamma_g\le 1/L$ and update $\theta_t^{+}=\theta_t-\gamma_g\tilde{g}_t$. Fix any $a\in(0,1)$ and define
    \[
    c_g(\gamma_g)\;\triangleq\;\tfrac{\gamma_g}{2a}+\tfrac{L\gamma_g^2}{2}.
    \]
    Then
    \[
    \!\left[\varphi(\theta_t^{+})-\varphi^\star\right]
    \;\le\;
    (1-\mu\gamma_g)\,\!\left[\varphi(\theta_t)-\varphi^\star\right]
    \;+\;
    \]
\begin{equation}
c_g(\gamma_g)\,\Big(\tfrac{c_s}{m}\sigma_g^2 + c_\tau\zeta^2 + [\beta_t^2]\Gamma^2\Big).
\label{eq:half}
\end{equation}
\end{proposition}
\begin{proof}
By \eqref{eq:descent} with $g=\tilde{g}_t$,
\begin{align}
\!\left[\varphi(\theta_t^{+})\right]
&\le \varphi(\theta_t)
- \gamma_g\,\langle \nabla\varphi(\theta_t), \tilde{g}_t\rangle
+ \tfrac{L\gamma_g^2}{2}\,\|\tilde{g}_t\|^2. \label{eq:half-1}
\end{align}
Write $\tilde{g}_t=\nabla\varphi(\theta_t)+e_t$ and expand the inner product:
\[
\langle \nabla\varphi(\theta_t), \tilde{g}_t\rangle
= \|\nabla\varphi(\theta_t)\|^2 + \langle \nabla\varphi(\theta_t), e_t\rangle.
\]
Apply Young’s inequality with $a\in(0,1)$:
\(
\langle \nabla\varphi(\theta_t), e_t\rangle
\ge -\tfrac{a}{2}\|\nabla\varphi(\theta_t)\|^2 - \tfrac{1}{2a}\|e_t\|^2.
\)
Also $\mathbb{E}[\|\tilde{g}_t\|^2] \le 2\|\nabla\varphi(\theta_t)\|^2 + 2\mathbb{E}[\|e_t\|^2]$.
Substitute into \eqref{eq:half-1}, use \eqref{eq:pl}, collect terms, and simplify with $\gamma_g\le 1/L$ to get
\[
\!\left[\varphi(\theta_t^{+})-\varphi^\star\right]
\le (1-\mu\gamma_g)\,\!\left[\varphi(\theta_t)-\varphi^\star\right]
+ \Big(\tfrac{\gamma_g}{2a}+\tfrac{L\gamma_g^2}{2}\Big)\|e_t\|^2.
\]
Finally apply Lemma~\ref{lem:delta}.
\end{proof}

\subsection{Reserve-set retraining}

\begin{proposition}[Reserve \(r\)-step descent]\label{prop:reserve}
Let $\gamma_r\le 1/L$ and run $r$ unbiased reserve steps from $\theta_t^{+}$ with variance $\sigma_r^2$ and mismatch $\varepsilon_r$. Then
\[
\!\left[\varphi(\theta_{t+1})-\varphi^\star \,\big|\, \theta_t^{+}\right]
\le (1-\mu\gamma_r)^{r}\,\big(\varphi(\theta_t^{+})-\varphi^\star\big)
\]
\begin{align}
+ \tfrac{L\,\gamma_r^2\,r}{2}\,\sigma_r^2
+ (1-\mu\gamma_r)^{r}\,\varepsilon_r^2 .
\label{eq:reserve}
\end{align}
\end{proposition}
\begin{proof}
Apply Lemma~\ref{lem:contract} sequentially to the clean reserve objective (contraction $(1-\mu\gamma_r)$ per step), add the variance term $\tfrac{L\gamma_r^2}{2}\sigma_r^2$ per step, and carry the fixed mismatch as a contracted bias $(1-\mu\gamma_r)^{r}\varepsilon_r^2$.
\end{proof}

\subsection{Composition and conclusion}

Condition on $\theta_t^{+}$ and apply \eqref{eq:reserve}, then take full expectation:
\[
\!\left[\varphi(\theta_{t+1}) - \varphi^\star\right]
\le (1-\mu\gamma_r)^{r}\,\!\left[\varphi(\theta_t^{+}) - \varphi^\star\right]\]
\begin{align}
   + \tfrac{L\,\gamma_r^{2}\,r}{2}\,\sigma_r^2
   + (1-\mu\gamma_r)^{r}\,\varepsilon_r^2,
\end{align}
\[
\le (1-\mu\gamma_r)^{r}
   \Big((1-\mu\gamma_g)\,\!\left[\varphi(\theta_t) - \varphi^\star\right]
   + C_{\mathrm{local}}'\Big)
\]
\begin{align}
   + \tfrac{L\,\gamma_r^{2}\,r}{2}\,\sigma_r^2
   + (1-\mu\gamma_r)^{r}\,\varepsilon_r^2, \label{eq:comp-2}
\end{align}
where
\begin{equation}
C_{\mathrm{local}}' \;=\; c_g(\gamma_g)\Big(\tfrac{c_s}{m}\sigma_g^2 + c_\tau\zeta^2 +[\beta_t^2]\Gamma^2\Big).
\label{eq:Clocal}
\end{equation}
Define
\[
q \;\triangleq\; (1-\mu\gamma_g)(1-\mu\gamma_r)^{r}, 
\qquad
\]
\begin{equation}
C' \;\triangleq\; (1-\mu\gamma_r)^{r}\,C_{\mathrm{local}}'
+ \tfrac{L\,\gamma_r^{2}\,r}{2}\,\sigma_r^2
+ (1-\mu\gamma_r)^{r}\,\varepsilon_r^2.
\label{eq:qC}
\end{equation}
Substituting \eqref{eq:qC} into \eqref{eq:comp-2} yields
\begin{equation}
\!\left[\varphi(\theta_{t+1}) - \varphi^\star\right]
\;\le\;
q\,\!\left[\varphi(\theta_t) - \varphi^\star\right] + C'.
\end{equation}
This is exactly the statement of Theorem~\ref{thm:conv-main}. \qed

\begin{IEEEbiography}
[{\includegraphics[width=1in,height=1.25in,clip,keepaspectratio]{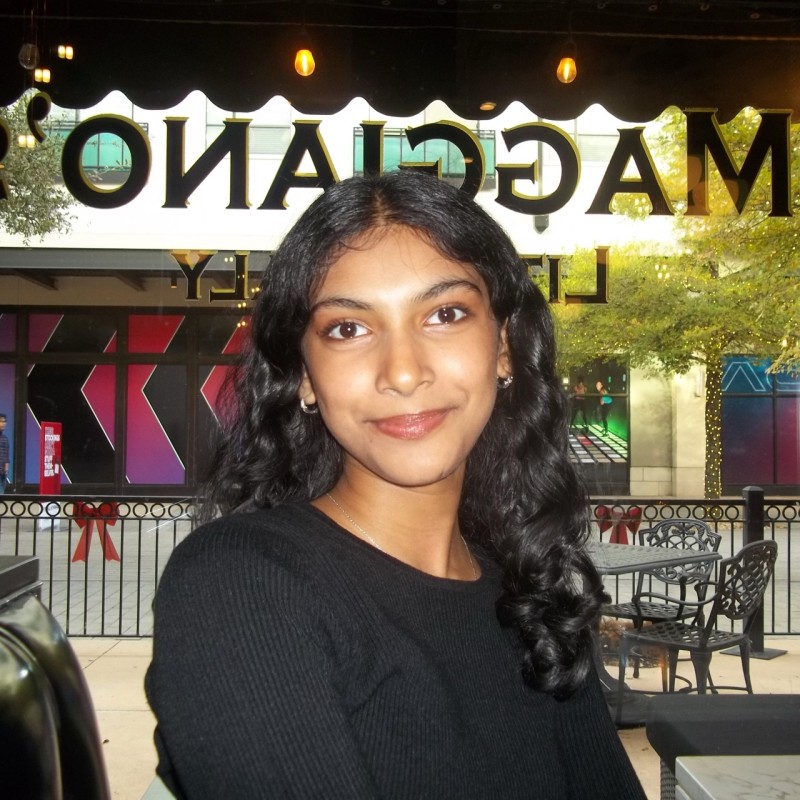}}]{Sathwika Peechara} is currently pursuing the B.S. in btoh Computer Science and Probability and Statistics at the University of California San Diego. Her research interests lie in AI cybersecurity, distributed ML systems, and federated learning.
\end{IEEEbiography}

\begin{IEEEbiography}
[{\includegraphics[width=1in,height=1.25in,clip,keepaspectratio]{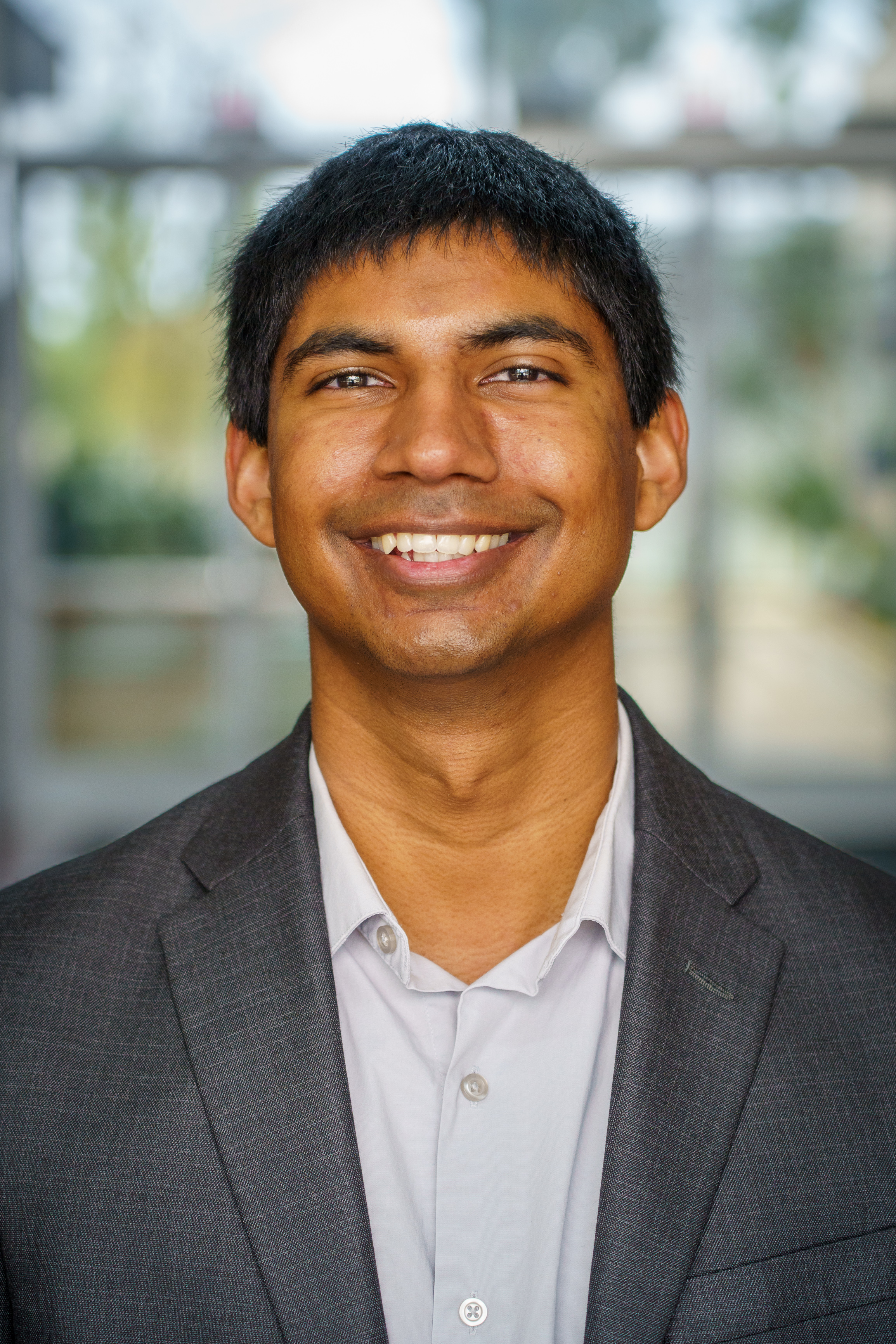}}]{Rajeev Sahay} received the B.S. degree in electrical engineering from The University of Utah, Salt Lake City, UT, USA, in 2018, and the M.S. and Ph.D. degrees in electrical and computer engineering from Purdue University, West Lafayette, IN, USA, in 2021 and 2022, respectively. Currently, he is a faculty member in the Department of Electrical and Computer Engineering at the University of California San Diego (UCSD). He is the recipient of the best undergraduate teaching award in the department of Electrical and Computer Engineering at UCSD, the Purdue Engineering Dean’s Teaching Fellowship and was named an Exemplary Reviewer by the IEEE Wireless Communications Letters. His research interests lie in the intersection of networking and machine learning, especially in their applications to signal processing, wireless communications, and engineering education. 
\end{IEEEbiography}

\end{document}